\documentclass[11pt,a4paper]{article}
\usepackage[margin=25mm]{geometry}
\usepackage[T1]{fontenc}
\usepackage{amsmath,amssymb,amsthm}
\usepackage[font=small,labelfont=bf]{caption}
\usepackage{booktabs}
\usepackage{array}
\usepackage{longtable}
\usepackage{placeins}
\usepackage{algorithm}
\usepackage{algpseudocode}
\usepackage{xcolor}
\usepackage{tikz}
\usetikzlibrary{matrix,calc,arrows.meta,fit,backgrounds,shapes.geometric}
\usepackage{pgfplots}
\pgfplotsset{compat=1.18, every axis/.append style={label style={font=\small}, tick label style={font=\small},
legend style={font=\small, draw=black!25, fill=white, fill opacity=0.9, text opacity=1}, legend cell align=left}, grid
style={black!12, line width=0.3pt}}
\usepackage[hidelinks]{hyperref}

\definecolor{tabblue}{HTML}{1F77B4}
\definecolor{taborange}{HTML}{FF7F0E}
\definecolor{tabgreen}{HTML}{2CA02C}
\definecolor{tabred}{HTML}{D62728}
\definecolor{tabpurple}{HTML}{9467BD}
\colorlet{lbcolor}{tabblue}
\colorlet{kcolor}{tabred}
\colorlet{optcolor}{tabgreen}
\colorlet{sparsecolor}{tabblue}
\colorlet{densecolor}{taborange}
\colorlet{termA}{tabblue}
\colorlet{termB}{taborange}
\colorlet{termC}{tabgreen}

\newcommand{\NInst}{375}
\newcommand{\NSparse}{190}
\newcommand{\NDense}{185}
\newcommand{\NProven}{208}
\newcommand{\NProvenSparse}{187}
\newcommand{\NProvenDense}{21}

\newcommand{\NOpenSparse}{3}

\newcommand{\NProvenReal}{197}
\newcommand{\NProvenSatOnly}{11}

\newcommand{\NRuns}{5}
\newcommand{\NRunsSame}{280}
\newcommand{\NRunsSameDenseOpen}{69}
\newcommand{\NRunsDiffer}{95}
\newcommand{\MaxRunSpread}{3}
\newcommand{\NRunsSingleBest}{39}
\newcommand{\DenseLargeRunStd}{0.6}
\newcommand{\NBetter}{278}
\newcommand{\NEqual}{79}

\newcommand{\NNew}{18}

\newcommand{\GroupTable}{%
\begin{tabular}{llrrrrrr}\toprule
class & $n$ & inst. & length & bound & proven & time to & running \\
 & & & $/n$ & $/n$ & & best [s] & time [s] \\ \midrule
sparse & 3--9 & 70 & 0.97 & 0.97 & 70 & $<$0.01 & $<$0.01 \\
sparse & 10--16 & 70 & 1.00 & 1.00 & 70 & $<$0.01 & $<$0.01 \\
sparse & 24--100 & 50 & 1.00 & 1.00 & 47 & 0.02 & 1.63 \\
\midrule
dense & 3--7 & 50 & 2.69 & 2.18 & 21 & 5.65 & 49 \\
dense & 8--16 & 90 & 4.43 & 1.79 & 0 & 155 & 660 \\
dense & 24--100 & 45 & 4.17 & 1.58 & 0 & 111 & 660 \\
\bottomrule\end{tabular}}
\newcommand{\StageTable}{%
\begin{tabular}{llrrr}\toprule
lower bound & valid for & sparse & dense & all \\ \midrule
degree & real weights & 23 & 0 & 23 \\
two-line bound & real weights & 164 & 10 & 174 \\
SAT certificate & precision $1/s$ & 0 & 11 & 11 \\
\midrule
proven optimal & & 187 & 21 & 208 \\
no proof (600\,s) & & 3 & 164 & 167 \\
\bottomrule\end{tabular}}
\newcommand{\MethodTable}{%
\begin{tabular}{lrrr}\toprule
component & sparse & dense, $n\le16$ & dense, $n\ge24$ \\ \midrule
Birkhoff algorithm, zeroing rule & 118 & 18 & 0 \\
frequency rule, unique sums & 47 & 0 & 0 \\
enumeration of weight sets & 20 & 1 & 0 \\
beam search & 0 & 104 & 0 \\
SAT re-optimization of the tail & 0 & 11 & 9 \\
large neighborhood search & 1 & 6 & 35 \\
randomized restarts & 4 & 0 & 1 \\
\bottomrule\end{tabular}}
\newcommand{\FamilyTable}{%
\begin{tabular}{lrrrrrrr}\toprule
family & $n$ & inst. & length & bound & proven & time to & running \\
 & & & & & & best [s] & time [s] \\ \midrule
\texttt{B3\_3} & 3 & 10 & 2.7 & 2.7 & 10 & $<$0.01 & $<$0.01 \\
\texttt{B4\_4} & 4 & 10 & 3.8 & 3.8 & 10 & $<$0.01 & $<$0.01 \\
\texttt{B5\_5} & 5 & 10 & 4.8 & 4.8 & 10 & $<$0.01 & $<$0.01 \\
\texttt{B6\_6} & 6 & 10 & 6.0 & 6.0 & 10 & $<$0.01 & $<$0.01 \\
\texttt{B7\_7} & 7 & 10 & 7.0 & 7.0 & 10 & $<$0.01 & $<$0.01 \\
\texttt{B8\_8} & 8 & 10 & 8.0 & 8.0 & 10 & $<$0.01 & $<$0.01 \\
\texttt{B9\_9} & 9 & 10 & 9.0 & 9.0 & 10 & $<$0.01 & $<$0.01 \\
\texttt{B10\_10} & 10 & 10 & 10.0 & 10.0 & 10 & $<$0.01 & $<$0.01 \\
\texttt{B11\_11} & 11 & 10 & 11.0 & 11.0 & 10 & $<$0.01 & $<$0.01 \\
\texttt{B12\_12} & 12 & 10 & 12.0 & 12.0 & 10 & $<$0.01 & $<$0.01 \\
\texttt{B13\_13} & 13 & 10 & 13.0 & 13.0 & 10 & $<$0.01 & $<$0.01 \\
\texttt{B14\_14} & 14 & 10 & 14.0 & 14.0 & 10 & $<$0.01 & $<$0.01 \\
\texttt{B15\_15} & 15 & 10 & 15.0 & 15.0 & 10 & $<$0.01 & $<$0.01 \\
\texttt{B16\_16} & 16 & 10 & 16.0 & 16.0 & 10 & $<$0.01 & $<$0.01 \\
\texttt{B24\_24} & 24 & 10 & 24.0 & 24.0 & 10 & $<$0.01 & $<$0.01 \\
\texttt{B32\_32} & 32 & 10 & 32.0 & 32.0 & 10 & $<$0.01 & 0.09 \\
\texttt{B48\_48} & 48 & 10 & 48.0 & 48.0 & 10 & $<$0.01 & 0.02 \\
\texttt{B64\_64} & 64 & 10 & 64.0 & 64.0 & 10 & 0.01 & 0.32 \\
\texttt{B100\_100} & 100 & 10 & 100.0 & 99.7 & 7 & 0.06 & 84 \\
\midrule
\texttt{B3\_9} & 3 & 10 & 4.5 & 4.5 & 10 & $<$0.01 & $<$0.01 \\
\texttt{B4\_16} & 4 & 10 & 9.6 & 9.6 & 10 & $<$0.01 & 0.21 \\
\texttt{B5\_25} & 5 & 10 & 13.9 & 12.9 & 1 & 3.82 & 589 \\
\texttt{B6\_36} & 6 & 10 & 19.2 & 13.7 & 0 & 18 & 661 \\
\texttt{B7\_49} & 7 & 10 & 24.9 & 15.0 & 0 & 139 & 660 \\
\texttt{B8\_64} & 8 & 10 & 29.9 & 15.5 & 0 & 53 & 660 \\
\texttt{B9\_81} & 9 & 10 & 34.0 & 16.3 & 0 & 126 & 660 \\
\texttt{B10\_100} & 10 & 10 & 46.5 & 18.0 & 0 & 179 & 660 \\
\texttt{B11\_121} & 11 & 10 & 51.7 & 20.0 & 0 & 187 & 660 \\
\texttt{B12\_144} & 12 & 10 & 55.8 & 21.3 & 0 & 289 & 660 \\
\texttt{B13\_169} & 13 & 10 & 60.8 & 23.0 & 0 & 194 & 660 \\
\texttt{B14\_196} & 14 & 10 & 64.9 & 24.6 & 0 & 226 & 660 \\
\texttt{B15\_225} & 15 & 10 & 68.1 & 26.0 & 0 & 233 & 660 \\
\texttt{B16\_256} & 16 & 10 & 72.7 & 27.7 & 0 & 76 & 660 \\
\texttt{B24\_576} & 24 & 10 & 98.0 & 39.0 & 0 & 72 & 660 \\
\texttt{B32\_1024} & 32 & 10 & 150.0 & 53.0 & 0 & 122 & 660 \\
\texttt{B48\_2304} & 48 & 10 & 199.5 & 75.0 & 0 & 127 & 660 \\
\texttt{B64\_4096} & 64 & 10 & 243.0 & 96.9 & 0 & 119 & 660 \\
\texttt{B100\_10000} & 100 & 5 & 404.2 & 151.0 & 0 & 142 & 660 \\
\bottomrule\end{tabular}}
\newcommand{\StatsTable}{%
\begin{tabular}{lrrrrr}\toprule
class & $n$ & scale $s$ & nonzeros per line & permutations per position & witness \\ \midrule
sparse & 3 & $10^{3}$ & 2.2 & 1.21 & 100\,\% \\
sparse & 4 & $10^{4}$ & 2.7 & 1.41 & 98\,\% \\
sparse & 5 & $10^{4}$ & 3.2 & 1.50 & 94\,\% \\
sparse & 6 & $10^{4}$ & 4.0 & 1.50 & 100\,\% \\
sparse & 7 & $10^{4}$ & 4.6 & 1.53 & 93\,\% \\
sparse & 8 & $10^{4}$ & 5.2 & 1.55 & 99\,\% \\
sparse & 9 & $10^{4}$ & 5.7 & 1.57 & 98\,\% \\
sparse & 10 & $10^{5}$ & 6.4 & 1.56 & 100\,\% \\
sparse & 11 & $10^{5}$ & 7.3 & 1.51 & 100\,\% \\
sparse & 12 & $10^{5}$ & 7.9 & 1.52 & 100\,\% \\
sparse & 13 & $10^{5}$ & 8.4 & 1.56 & 100\,\% \\
sparse & 14 & $10^{5}$ & 9.0 & 1.55 & 99\,\% \\
sparse & 15 & $10^{5}$ & 9.8 & 1.53 & 100\,\% \\
sparse & 16 & $10^{5}$ & 10.1 & 1.58 & 100\,\% \\
sparse & 24 & $10^{5}$ & 15.1 & 1.59 & 100\,\% \\
sparse & 32 & $10^{6}$ & 20.5 & 1.56 & 100\,\% \\
sparse & 48 & $10^{6}$ & 30.5 & 1.57 & 100\,\% \\
sparse & 64 & $10^{6}$ & 40.5 & 1.58 & 100\,\% \\
sparse & 100 & $10^{7}$ & 63.3 & 1.58 & 100\,\% \\
\midrule
dense & 3 & $10^{3}$ & 2.9 & 1.70 & 54\,\% \\
dense & 4 & $10^{4}$ & 4.0 & 2.91 & 14\,\% \\
dense & 5 & $10^{4}$ & 5.0 & 4.56 & 3\,\% \\
dense & 6 & $10^{4}$ & 6.0 & 5.88 & 1\,\% \\
dense & 7 & $10^{4}$ & 7.0 & 7.00 & 0\,\% \\
dense & 8 & $10^{4}$ & 8.0 & 8.00 & 0\,\% \\
dense & 9 & $10^{4}$ & 9.0 & 8.98 & 0\,\% \\
dense & 10 & $10^{5}$ & 10.0 & 10.00 & 0\,\% \\
dense & 11 & $10^{5}$ & 11.0 & 10.98 & 0\,\% \\
dense & 12 & $10^{5}$ & 12.0 & 12.00 & 0\,\% \\
dense & 13 & $10^{5}$ & 13.0 & 12.99 & 0\,\% \\
dense & 14 & $10^{5}$ & 14.0 & 13.99 & 0\,\% \\
dense & 15 & $10^{5}$ & 15.0 & 15.00 & 0\,\% \\
dense & 16 & $10^{5}$ & 16.0 & 15.99 & 0\,\% \\
dense & 24 & $10^{5}$ & 24.0 & 23.93 & 0\,\% \\
dense & 32 & $10^{6}$ & 32.0 & 31.99 & 0\,\% \\
dense & 48 & $10^{6}$ & 48.0 & 47.96 & 0\,\% \\
dense & 64 & $10^{6}$ & 64.0 & 63.85 & 0\,\% \\
dense & 100 & $10^{7}$ & 100.0 & 99.95 & 0\,\% \\
\bottomrule\end{tabular}}
\newcommand{\DenseRatioFive}{2.78}

\newcommand{\DenseLbRatioLarge}{1.58}
\newcommand{\DenseRatioMidMin}{3.8}
\newcommand{\DenseRatioMidMax}{4.7}
\newcommand{\DenseRatioBigMin}{3.8}
\newcommand{\DenseRatioBigMax}{4.7}
\newcommand{\DenseGapFactor}{2.6}
\newcommand{\SparseMedianTts}{$<$0.01}
\newcommand{\SparseMaxTts}{0.07}
\newcommand{\SparseProvenFast}{179}
\newcommand{\NProvenSparseReal}{187}
\newcommand{\CompTable}{%
\begin{tabular}{lrrrrrrr}\toprule
method & $n$ & exact & shorter & equal & longer & ratio & proven \\ \midrule
Birkhoff+~\cite{valls2021} & 3--100 & 360 & 301 & 59 & 0 & 2.06 & 0 / 208 \\
MIP~\cite{qoblibsub} & 3--6 & 80 & 10 & 70 & 0 & 1.02 & 66 / 61 \\
E-FCFW~\cite{qoblibsub} & 3--8 & 48 & 5 & 43 & 0 & 1.05 & 0 / 48 \\
quantum black box~\cite{qoblibsub} & 4--5 & 24 & 4 & 20 & 0 & 1.03 & 0 / 24 \\
partitioned QAOA~\cite{qoblibsub} & 16 & 1 & 1 & 0 & 0 & 4.12 & 0 / 1 \\
\bottomrule\end{tabular}}
\newcommand{\CompBplusN}{360}
\newcommand{\CompBplusRatio}{2.1}
\newcommand{\CompBplusTime}{0.03}
\newcommand{\CompBplusOurTime}{$<$0.01}
\newcommand{\CompMipN}{80}
\newcommand{\CompMipProven}{66}
\newcommand{\CompMipOurProven}{61}

\newcommand{\CompEfcfwN}{48}

\newcommand{\CompQbN}{24}
\newcommand{\CompQbShorter}{4}

\newcommand{\CompEfcfwShorter}{5}
\newcommand{\CompEfcfwEqual}{43}
\newcommand{\FreqInstance}{B16\_16\_1}
\newcommand{\FreqDistinct}{73}
\newcommand{\FreqWeights}{16}
\newcommand{\ReachTwoMid}{0.4}
\newcommand{\ReachOneMid}{8.2}
\newcommand{\ReachTwoTwentyFour}{2.5}
\newcommand{\ReachOneTwentyFour}{17.7}
\newcommand{\ReachTwoThirtyTwo}{3.1}
\newcommand{\ReachOneThirtyTwo}{14.3}
\newcommand{\FigBoundsXmax}{499.0}
\newcommand{\FigBoundsSeparator}{251.5}
\newcommand{\FigBoundsTicks}{7,33,59,85,111,137,163,189,215,241,263,289,315,341,367,393,419,445,471,494.5}
\newcommand{\FigBoundsTickLabels}{3,5,7,9,11,13,15,24,48,100,3,5,7,9,11,13,15,24,48,100}
\newcommand{\FigBoundsTicksLow}{20,46,72,98,124,150,176,202,228,276,302,328,354,380,406,432,458,484}
\newcommand{\FigBoundsTickLabelsLow}{4,6,8,10,12,14,16,32,64,4,6,8,10,12,14,16,32,64}
\newcommand{\FigBoundsMidSparse}{124}
\newcommand{\FigBoundsMidDense}{377.5}

\newcommand{\Z}{\mathbb{Z}}
\newcommand{\R}{\mathbb{R}}

\newtheorem{theorem}{Theorem}
\newtheorem{lemma}[theorem]{Lemma}
\newtheorem{proposition}[theorem]{Proposition}

\begin{document}

\title{Exact minimum Birkhoff decompositions of doubly stochastic matrices with finite precision}
\author{Andr\'as Cz\'egel\thanks{Bentax Plus}}
\date{}
\maketitle

\begin{abstract}
The minimum Birkhoff decomposition problem asks for a representation of a doubly stochastic matrix as a convex
combination of as few permutation matrices as possible. This NP-hard problem arises in the scheduling of switches in
data centers and satellite systems, in random assignments in economics and in the statistics of rankings. We
consider doubly stochastic matrices whose entries are given with a fixed number of decimal digits and look for exact
decompositions whose weights have the same precision. We propose a parallel method with three parts. For matrices
built from few permutation matrices, the first part reads the weights from the frequencies of the entries and places
them on the matrix positions by a rule that, given the correct weights, never fixes a position wrongly. The second
part computes lower bounds from
pairs of rows or columns, which remain valid for real weights, and from certificates of a satisfiability solver. For
matrices built from many permutation matrices, the third part runs a beam search over steps that each make at least one
entry zero. On \NInst{} test instances with up to $100$ rows, the method proves optimality for \NProven{} instances, and
\NProvenReal{} of these proofs also hold for real weights. Of the \NSparse{} instances built from few permutation
matrices, \NProvenSparse{} are solved to proven optimality, \SparseProvenFast{} of them within one second. Compared with
the best exact decompositions known before, the new decompositions are shorter for \NBetter{} instances and never
longer.
\end{abstract}

\section{Introduction}\label{sec:introduction}

A nonnegative $n\times n$ matrix is \emph{doubly stochastic} if all its row and column sums are equal to one, and a
\emph{permutation matrix} is a $0/1$ matrix with exactly one entry equal to one in every row and every column. The
theorem of Birkhoff and von Neumann states that every doubly stochastic matrix $A$ is a convex combination
$A=\sum_{l=1}^k\lambda_lP_l$ of permutation matrices $P_l$, with weights $\lambda_l>0$ and
$\sum_l\lambda_l=1$~\cite{birkhoff1946,vonneumann1953}. Such a representation is a Birkhoff--von Neumann (BvN)
decomposition, and $k$ is its \emph{length}. The constructive proof of the theorem gives the \emph{Birkhoff algorithm}:
choose a permutation matrix whose ones lie at positions of nonzero entries of the matrix, subtract it with the largest
weight that keeps the matrix nonnegative, and repeat with the remaining matrix.

A doubly stochastic matrix usually has many BvN decompositions, which can differ both in their permutation matrices and
in their length. Since the set of doubly stochastic $n\times n$ matrices, the \emph{Birkhoff polytope}, has dimension
$(n-1)^2$, Carath\'eodory's theorem~\cite{caratheodory1911} guarantees a decomposition of length at most
$(n-1)^2+1$~\cite{marcus1959}. The \emph{minimum Birkhoff decomposition problem} asks for a BvN decomposition of minimum
length. In this paper, we propose a method that computes short exact decompositions and lower bounds on the minimum
length. If both are equal, the decomposition is optimal.

\subsection{Applications}

The main application is the scheduling of switches. A switch with $n$ inputs and $n$ outputs connects every input to
exactly one output at a time, so each of its configurations is a permutation. If the matrix of traffic demands is scaled
to a doubly stochastic matrix, a BvN decomposition is a schedule in which the switch uses configuration $P_l$ for the
fraction $\lambda_l$ of the time. Chang, Chen and Huang~\cite{chang2000} applied this idea to the switches of computer
networks, and Ye et al.~\cite{ye2017} extended such switches by deflection compensation. In data centers, optical
switches transmit data much faster than electronic ones, but every reconfiguration takes between microseconds and
milliseconds~\cite{farrington2010,porter2013}, during which no data can be sent, so schedules with few configurations
save time. Towles and Dally~\cite{towles2003} and Liu et al.~\cite{liu2015} schedule switches with a time cost for every
reconfiguration (see also~\cite{bojja2016,valls2021,tan2021}). Kulkarni, Lee and Singh~\cite{kulkarni2017} name routing
in data centers with reconfigurable switches as the main motivation for the minimum BvN decomposition problem.

The same problem arose earlier in satellite communication. In satellite-switched time division multiple access (SS/TDMA)
systems, a satellite connects ground stations according to a sequence of switching modes, each of which is a
permutation. The traffic matrix is decomposed into a weighted sum of switching modes, with the aim of a short total
transmission time, few switching modes, or both~\cite{balas1983,burkard1985}. Rendl~\cite{rendl1985} and Gopal and
Wong~\cite{gopal1985} showed that minimizing the number of switching modes is NP-hard for variants of this problem.

BvN decompositions also appear in economics, statistics and quantum computing. In economics, a random assignment of $n$
objects to $n$ agents is described by the doubly stochastic matrix of the probabilities that agent $i$ receives object
$j$, and a probability distribution over deterministic assignments with these probabilities is a BvN
decomposition~\cite{budish2013}. In statistics, a probability distribution over the rankings of $n$ items is summarized
by the doubly stochastic matrix of the probabilities that item $i$ has rank $j$, and a distribution with few rankings
that has this matrix is a short BvN decomposition~\cite{jagabathula2011}. In quantum computing, the squared moduli
$|U_{ij}|^2$ of the entries of a unitary matrix $U$ form a doubly stochastic matrix, and Mariella et
al.~\cite{mariella2024} use this to represent transport plans. The Quantum Optimization Benchmarking
Library~\cite{koch2025} contains the minimum Birkhoff decomposition problem as a test problem for quantum and classical
methods.

\subsection{Related work}

Dufoss\'e and U\c{c}ar~\cite{dufosse2016} showed that the minimum BvN decomposition problem is NP-hard. Kulkarni, Lee
and Singh~\cite{kulkarni2017} studied the variant in which the matrix is written as a combination of sub-permutation
matrices, that is, of matchings that need not be perfect. They showed that this variant is NP-hard even when the optimal
length is a constant $k\ge4$, so it is not fixed-parameter tractable in $k$ unless $\mathrm{P}=\mathrm{NP}$, and they
gave an approximation algorithm with ratio $O(\log k)$. Brualdi~\cite{brualdi1982} and Johnson, Dulmage and
Mendelsohn~\cite{johnson1960} bounded the number of permutation matrices that the Birkhoff algorithm uses in terms of
the nonzero pattern. We know of no lower bound in this literature beyond the largest number of nonzero entries in a row
or a column, which we call the degree in Section~\ref{sec:problem}.

Most methods used in practice are variants of the Birkhoff algorithm that differ in the choice of the permutation matrix
in each step. The usual choice is the permutation matrix with the largest possible weight, the \emph{bottleneck}
permutation, found by a bottleneck matching algorithm~\cite{gabow1988,duff2001,panagiotas2023}. Dufoss\'e and
U\c{c}ar~\cite{dufosse2016} studied this rule, and Dufoss\'e, Kaya, Panagiotas and U\c{c}ar~\cite{dufosse2018} showed
that every heuristic of this type can miss the optimum, as Brualdi~\cite{brualdi1982} had suggested. Liu et
al.~\cite{liu2015} use thresholds to choose the permutation matrices. Valls, Iosifidis and Tassiulas~\cite{valls2021}
bound the length that the Birkhoff algorithm produces for approximate decompositions and propose Birkhoff+, which
chooses each permutation matrix by a linear program inspired by the Frank--Wolfe method. In the TDMA literature,
Inukai~\cite{inukai1979} gave an algorithm with few switching modes, Rote and Vogel~\cite{rote1993} proposed a heuristic
based on scaling and edge coloring, and Prais and Ribeiro~\cite{prais2000} used a reactive greedy randomized adaptive
search procedure (GRASP).

Exact methods start with Minoux~\cite{minoux1986}, who formulated the assignment of a traffic matrix as a set covering
problem over switching modes and solved its linear relaxation by column generation. Ribeiro, Minoux and
Penna~\cite{ribeiro1989} built an exact column generation method on this formulation. Mixed-integer programs with one
binary variable per permutation matrix, or with a fixed number of slots for permutation matrices, solve small instances,
and Koch et al.~\cite{koch2025} report results of CPLEX for such a model. Hendrych et al.~\cite{hendrych2025} combine
branch-and-bound with Frank--Wolfe methods for convex mixed-integer problems, which include sparse approximate
decompositions over the Birkhoff polytope.

Frank--Wolfe (conditional gradient) methods produce sparse convex combinations of the vertices of a polytope. For the
Birkhoff polytope, the vertices are the permutation matrices, and the linear subproblem is an assignment
problem~\cite{jaggi2013,braun2019,combettes2023}. These methods are fast and give good approximate decompositions but
rarely reach an exact one: among the results submitted to the library~\cite{qoblibsub}, an extended Frank--Wolfe method
gives exact decompositions for \CompEfcfwN{} of its $131$ instances (Section~\ref{sec:comparison}).

For the recovery of sparse mixtures, the sparsest-fit algorithm of Jagabathula and Shah~\cite{jagabathula2011} processes
the entries in increasing order and decides for each entry whether it is a new weight or a sum of known weights. They
showed that a weighted sum of few random permutation matrices can be decomposed exactly in this way if two conditions
hold: every permutation matrix has a one at a position where no other permutation matrix has a one, a \emph{witness
position}, and no two different subsets of the weights have the same sum.

\subsection{Contributions}

The entries of our test instances have $3$ to $7$ decimal digits, and we look for exact decompositions whose weights
have the same precision. The mixed-integer model of Koch et al.~\cite{koch2025} uses the same precision, and in the TDMA
literature the weights are numbers of time slots~\cite{inukai1979,rote1993}. After scaling, all computations can be
carried out in exact integer arithmetic. Real weights can only allow shorter decompositions, so this matters only for
the lower bounds. For every bound, we state whether it also holds for real weights.

Our method rests on one observation: in every decomposition, each row and each column of the matrix splits the list of
weights into groups whose sums are the entries of that row or column. We call the method RCS, after its three parts
Recover, Certify and Search. Recover places the weights in matrices built from few permutation matrices, where most
weights appear directly as entries. Certify proves lower bounds for all matrices. Search builds decompositions step by
step in matrices built from many permutation matrices, where no weight is visible. The parts run in parallel, and a run
stops as soon as the best decomposition reaches the lower bound. The contributions of this paper are the following.
\begin{itemize}
\item We give a rule, fixing by unique subset sums, that places known weights on the positions of the matrix and, if the
  weights are correct, fixes every position correctly.
\item We derive a lower bound from two rows or two columns that is valid for real weights. It is computed exactly when
  few values remain after equal values have been paired, and it is replaced by a weaker bound from a linear program
  otherwise. We also give a model for the Boolean satisfiability problem (SAT) that decides whether $k$ permutation
  matrices with weights of precision $1/s$ suffice. For small matrices built from many permutation matrices, we split
  this question into many small SAT problems, many of which are removed by a lower bound on the number of terms at each
  position.
\item We propose a beam search over zeroing steps, where a zeroing step subtracts a permutation matrix with the largest
  weight that keeps all entries nonnegative, so that at least one entry becomes zero. For $n\le7$, the search considers
  every permutation matrix in each step. We also give a tie-break rule for the Birkhoff algorithm that removes as many
  binary digits of the remaining entries as possible and shortens the decompositions of dense matrices with $n\ge32$ by
  $3$ to $6\,\%$.
\item We report a computational study on \NInst{} test instances with $n$ up to $100$, carried out with an
  implementation in Rust, including a comparison with the exact decompositions that other groups have submitted to the
  library for these instances.
\end{itemize}

Appendix~\ref{app:proofs} contains the proofs, Appendix~\ref{app:instances} describes the test instances, and
Appendix~\ref{app:results} gives the results per family of instances.

\section{Problem and basic properties}\label{sec:problem}

\subsection{Decompositions with finite precision}

Let $A$ be a doubly stochastic $n\times n$ matrix whose entries are multiples of $1/s$ for some integer $s\ge1$. For
data with $d$ decimal digits, $s=10^d$. We work with the scaled matrix $D=sA\in\Z_{\ge0}^{n\times n}$, whose row and
column sums are all equal to $s$. For a permutation $\pi$ of $\{1,\dots,n\}$, let $P_\pi$ denote the permutation matrix
with ones at the positions $(i,\pi(i))$. A \emph{decomposition} of $D$ of length $k$ is a list $(c_l,\pi_l)_{l=1}^k$ of
weights $c_l>0$ and permutations $\pi_l$ such that $D=\sum_{l=1}^kc_lP_{\pi_l}$, and the pairs $(c_l,\pi_l)$ are its
\emph{terms} (Figure~\ref{fig:example}). It follows that $\sum_lc_l=s$, and the weights $\lambda_l=c_l/s$ give the BvN
decomposition $A=\sum_l\lambda_lP_{\pi_l}$. A decomposition has \emph{precision $1/s$} if all $c_l$ are integers, that
is, if all weights $\lambda_l$ are multiples of $1/s$, like the entries of $A$. The scaling does not change the problem.
It only allows exact integer arithmetic without rounding errors. In what follows, the weights are the numbers $c_l$,
that is, the weights $\lambda_l$ in units of $1/s$. We denote by $k^*(D)$ the minimum length of a decomposition with
precision $1/s$ and by $k^*_{\R}(D)$ the minimum length of a decomposition with real weights, so that
\[ k^*_{\R}(D)\;\le\;k^*(D). \]
RCS searches for decompositions with precision $1/s$, and \emph{optimal} means of length $k^*(D)$. We do not know
whether $k^*_{\R}(D)$ and $k^*(D)$ differ for any of our test instances.

A \emph{line} is a row or a column, and the \emph{degree} $\deg(D)$ is the largest number of nonzero entries in a line.
A term $(c_l,\pi_l)$ \emph{covers} the $n$ positions $(i,\pi_l(i))$. If some terms have been chosen such that no entry
becomes negative, the \emph{residual} $R=D-\sum_{\text{chosen}}c_lP_{\pi_l}$ is the part of $D$ that remains to be
decomposed. It is again a nonnegative matrix with equal line sums, so all statements about $D$ also hold for residuals.

\begin{figure}[t]\centering
\begin{tikzpicture}[
  mat/.style={matrix of math nodes, left delimiter=(, right delimiter=), inner sep=0.3pt, column sep=0pt,
    row sep=0pt, nodes={minimum width=1.15em, minimum height=1.45em, anchor=center, inner sep=0pt}},
  tA/.style={text=termA, font=\boldmath}, tB/.style={text=termB, font=\boldmath},
  tC/.style={text=termC, font=\boldmath}, zero/.style={text=black!40},
  dA/.style={circle, fill=termA, inner sep=0pt, minimum size=2.4pt},
  dB/.style={rectangle, fill=termB, inner sep=0pt, minimum size=2.2pt},
  dC/.style={regular polygon, regular polygon sides=3, fill=termC, inner sep=0pt, minimum size=3.8pt}]
\matrix[mat] (D) {5 & 5 & |[zero]|0 & |[zero]|0\\ 3 & 5 & 2 & |[zero]|0\\ |[zero]|0 & |[zero]|0 & 8 & 2\\
  2 & |[zero]|0 & |[zero]|0 & 8\\};
\node[anchor=east] at ([xshift=-5pt]D.west) {$D=$};
\node[anchor=west] (e1) at ([xshift=5pt]D.east) {$=\textcolor{termA}{5}\,\cdot$};
\matrix[mat, anchor=west] (P1) at ([xshift=5pt]e1.east)
  {|[tA]|1 & |[zero]|0 & |[zero]|0 & |[zero]|0\\ |[zero]|0 & |[tA]|1 & |[zero]|0 & |[zero]|0\\
   |[zero]|0 & |[zero]|0 & |[tA]|1 & |[zero]|0\\ |[zero]|0 & |[zero]|0 & |[zero]|0 & |[tA]|1\\};
\node[anchor=west] (e2) at ([xshift=5pt]P1.east) {$+\textcolor{termB}{3}\,\cdot$};
\matrix[mat, anchor=west] (P2) at ([xshift=5pt]e2.east)
  {|[zero]|0 & |[tB]|1 & |[zero]|0 & |[zero]|0\\ |[tB]|1 & |[zero]|0 & |[zero]|0 & |[zero]|0\\
   |[zero]|0 & |[zero]|0 & |[tB]|1 & |[zero]|0\\ |[zero]|0 & |[zero]|0 & |[zero]|0 & |[tB]|1\\};
\node[anchor=west] (e3) at ([xshift=5pt]P2.east) {$+\textcolor{termC}{2}\,\cdot$};
\matrix[mat, anchor=west] (P3) at ([xshift=5pt]e3.east)
  {|[zero]|0 & |[tC]|1 & |[zero]|0 & |[zero]|0\\ |[zero]|0 & |[zero]|0 & |[tC]|1 & |[zero]|0\\
   |[zero]|0 & |[zero]|0 & |[zero]|0 & |[tC]|1\\ |[tC]|1 & |[zero]|0 & |[zero]|0 & |[zero]|0\\};
\foreach \r/\c/\mk/\dx in {1/1/dA/0, 1/2/dB/-2.3, 1/2/dC/2.3, 2/1/dB/0, 2/2/dA/0, 2/3/dC/0,
    3/3/dA/-2.3, 3/3/dB/2.3, 3/4/dC/0, 4/1/dC/0, 4/4/dA/-2.3, 4/4/dB/2.3}
  \node[\mk] at ($(D-\r-\c.south)+(\dx pt,1.8pt)$) {};
\node[dA, scale=1.3] at ([yshift=4pt]P1.north) {};
\node[dB, scale=1.3] at ([yshift=4pt]P2.north) {};
\node[dC, scale=1.3] at ([yshift=4pt]P3.north) {};
\end{tikzpicture}
\caption{A decomposition of a $4\times4$ matrix $D$, a doubly stochastic matrix scaled by $10$, into three terms with
the weights $5$, $3$ and $2$. Each term has its own marker, shown above its permutation matrix. Below each entry of $D$,
there is one marker for every term whose permutation matrix has a one at this position. Row~1 splits the weights into
the groups $\{5\}$ and $\{3,2\}$, and row~3 into $\{5,3\}$ and $\{2\}$. Row~2 has three nonzero entries, so every
decomposition needs at least three terms, and this one is optimal.}\label{fig:example}
\end{figure}

The following lemma is the basis of our method. Its proof, like all other proofs, is given in Appendix~\ref{app:proofs}.

\begin{lemma}[line partitions]\label{lem:partition}
Let $(c_l,\pi_l)_{l=1}^k$ be a decomposition of $D$ with real weights $c_l>0$, and let $I_{ij}=\{l:\pi_l(i)=j\}$ for
every position $(i,j)$. Then $D_{ij}=\sum_{l\in I_{ij}}c_l$ for every position $(i,j)$, and $I_{ij}=\emptyset$ if and
only if $D_{ij}=0$. For every row $i$, the sets $I_{ij}$, $j=1,\dots,n$, with $D_{ij}>0$ form a partition of
$\{1,\dots,k\}$, and for every column $j$, the sets $I_{ij}$, $i=1,\dots,n$, with $D_{ij}>0$ do the same. In particular,
the weights sum to $s$, and every decomposition has at least $\deg(D)$ terms, so that $k^*(D)\ge k^*_{\R}(D)\ge\deg(D)$.
\end{lemma}

A decomposition with precision $1/s$ always exists. In the bipartite multigraph with the rows on one side, the columns
on the other side and $D_{ij}$ parallel edges between row $i$ and column $j$, every vertex has degree $s$, so by
K\H{o}nig's edge coloring theorem~\cite{konig1916} the edges split into $s$ perfect matchings, each of which is a
permutation. In these terms, a decomposition with precision $1/s$ and length $k$ corresponds to a coloring of the edges
with $s$ colors in which every color class is a perfect matching and at most $k$ different perfect matchings occur, with
the matching $\pi_l$ used for $c_l$ colors.

\subsection{Sparse and dense matrices}

Two classes of matrices require different methods. We call the permutations from which a test matrix was built its
\emph{generating permutations}. A matrix built from few generating permutations, for example from $n$ random
permutations with random weights, is \emph{sparse} in our sense. Almost every weight then appears as an entry in about
$37\,\%$ of the rows, and $k=n$ is optimal for most such matrices. A matrix built from many generating permutations, for
example from $n^2$ of them, is \emph{dense}. Every entry is then a sum of many weights, no weight is visible, and the
minimum length is much smaller than the number of generating permutations. RCS does not know the generating permutations
and \emph{treats} a matrix as sparse if at least $45\,\%$ of its nonzero entries have a value that occurs more than once
in the matrix, and as dense otherwise. In Section~\ref{sec:method}, sparse and dense refer to this criterion, and in
Section~\ref{sec:results} and the appendices to the class of the instance generator. The two agree for all test instances
except four with $n\le4$, and the threshold is not critical, since every threshold between $41\,\%$ and $51\,\%$ gives
the same classification.

\section{The method RCS}\label{sec:method}

\begin{algorithm}[t]
\caption{One run of RCS.}\label{alg:overview}
\begin{algorithmic}[1]
\Require scaled matrix $D=sA$ (integer entries, all line sums equal to $s$), time limit
\State $U\gets$ Birkhoff algorithm with the zeroing rule on $D$ \Comment{Section~\ref{sec:search}}
\State $L\gets\deg(D)$ \Comment{lower bound}
\While{$U$ has more than $L$ terms and time is left} \Comment{in parallel}
\State Recover ($D$ treated as sparse): find the weights, fix them by unique subset sums (Section~\ref{sec:recover})
\State Certify: raise $L$ by the two-line bound and by SAT certificates (Section~\ref{sec:certify})
\State Search ($D$ treated as dense): beam search, SAT re-optimization of the tail (Section~\ref{sec:search})
\State further threads: randomized restarts, large neighborhood search (Section~\ref{sec:search})
\EndWhile
\State \Return $U$, $L$ and the time of every improvement of $U$
\end{algorithmic}
\end{algorithm}

\subsection{Overview}\label{sec:overview}

A run of RCS has a time limit and consists of two \emph{chains} that work in parallel (Algorithm~\ref{alg:overview}),
each of which is a group of threads that run one algorithm each. The \emph{solution chain} consists of Recover and
Search. Its threads offer decompositions to a shared best decomposition $U$, and every offer is checked in integer
arithmetic before it is accepted. The \emph{bound chain} consists of Certify, which raises a shared lower bound $L$ on
$k^*(D)$ step by step. The run stops as soon as $U$ has length $L$, because $U$ is then optimal. Otherwise, it stops at
the time limit and reports $U$ and $L$.

\subsection{Recover: sparse matrices}\label{sec:recover}

Recover first finds the weights and then places them on the positions of $D$.

\subsubsection*{Frequency rule}
If the matrix is built from $n$ random permutations, a given one of them is the only one
through its position in a given row with probability $(1-1/n)^{n-1}\approx\mathrm{e}^{-1}$. Its weight then appears as
an entry about $n/\mathrm{e}$ times, whereas a sum of two weights appears only at the positions where the two
permutations meet, so the weights are the most frequent values (Figure~\ref{fig:freq}). We sort the distinct values of the
entries by frequency, most frequent first, and take the shortest prefix of this list whose sum is exactly $s$, provided
that the prefix ends where the frequency drops. A value that occurs about twice as often as a typical weight counts
twice, because two weights can be equal. If this rule fails, we scan the values in increasing order, as the sparsest-fit
algorithm~\cite{jagabathula2011} does, and take a value as a new weight unless it is the sum of at most three known
weights. As long as the best decomposition has more than $n$ terms, a depth-first \emph{enumeration of weight sets} also
tries other weight sets, in order of frequency.

\begin{figure}[t]\centering
\begin{tikzpicture}
\begin{axis}[width=\linewidth, height=4.4cm, ybar, bar width=1.2pt, bar shift=0pt, ymin=0, ymax=11.5, xmin=0,
  xmax=\FreqDistinct+1, xtick=\empty, ytick={0,2,...,10}, xlabel={distinct values of the nonzero entries, sorted by
  frequency}, ylabel={occurrences}, ymajorgrids, legend style={at={(0.98,0.95)}, anchor=north east},
  legend image code/.code={\fill[#1] (0cm,-0.1cm) rectangle (0.25cm,0.1cm);}]
\addplot[fill=lbcolor, draw=none] table[x=rank, y=count, restrict expr to domain={\thisrow{weight}}{1:1}]
  {fig_freq.csv};
\addlegendentry{weight of the decomposition}
\addplot[fill=black!30, draw=none] table[x=rank, y=count, restrict expr to domain={\thisrow{weight}}{0:0}]
  {fig_freq.csv};
\addlegendentry{other value (a sum of weights)}
\end{axis}
\end{tikzpicture}
\caption{Frequencies of the values in the sparse test instance \FreqInstance, a $16\times16$ matrix built from $16$
random permutation matrices. The bars show the \FreqDistinct{} distinct values of the nonzero entries, sorted by their
number of occurrences. The \FreqWeights{} weights of the decomposition found by RCS (dark bars) are exactly the
\FreqWeights{} most frequent values. All other values (light bars) are sums of weights and occur at most three
times.}\label{fig:freq}
\end{figure}
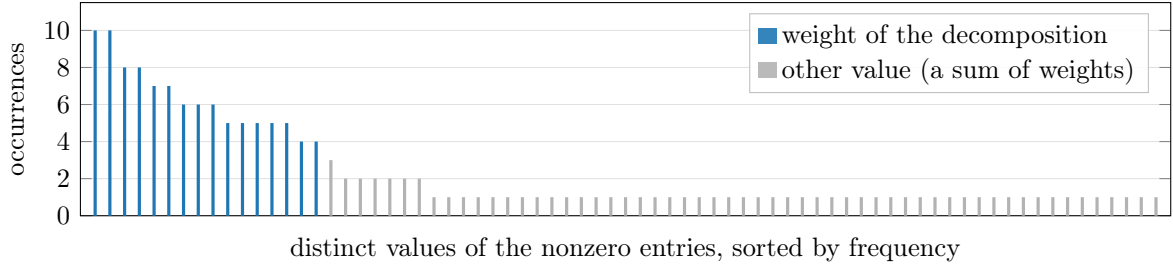

\subsubsection*{Fixing by unique subset sums}
Let the weights $c_1\le\dots\le c_k$ be given, and let
$F_i,G_j\subseteq\{1,\dots,k\}$ be the index sets of the weights that are not yet placed in row $i$ and in column $j$,
where equal weights count as different indices. A position is \emph{open} until its set of weights is fixed. We visit
the open positions in increasing order of their entries and, for a position $(i,j)$, list the subsets of $F_i\cap G_j$
with sum $D_{ij}$, stopping after two. If there is exactly one such subset, the position is fixed to it, and its indices
are removed from $F_i$ and $G_j$. If there is none, the weights do not fit $D$. If there are two, or if the enumeration
exceeds a fixed budget, the position stays open. These passes are repeated until no position changes, and a depth-first
search decides the positions that are still open. If the weights are correct, every fixed position is correct
(Theorem~\ref{thm:sweep}). The increasing order keeps the lists short, because the positions with smaller entries in the
row and the column of a visited position are usually fixed already, so only a few small weights remain to be combined.

\subsubsection*{Largest weight first}
A faster alternative places the weights one at a time, largest first, by solving one
assignment problem per weight. Since this procedure can fail, its result is accepted only if the final residual is zero.

\subsection{Certify: lower bounds}\label{sec:certify}

While the solution chain shortens the best decomposition, Certify raises the lower bound $L$ on $k^*(D)$ in three
stages: the degree (Lemma~\ref{lem:partition}), the two-line bound and SAT certificates. The degree and the two-line
bound depend only on the entries of $D$ and hold for real weights, so they also bound $k^*_{\R}(D)$, whereas the SAT
certificates hold only for weights with precision $1/s$.

\subsubsection*{Two-line bound}
Take two lines, and let $a$ and $b$ be the multisets of their nonzero entries, with $p$ and
$q$ elements. Every term covers one position in each of the two lines and thus links these two positions. Positions that
are linked, directly or through other positions, form groups in which the entries of the two lines have the same sum,
and we call any set of positions of the two lines with this property a \emph{block}. Since the linked groups are blocks
and a group with $r$ positions needs at least $r-1$ links,
\[ k^*(D)\;\ge\;k^*_{\R}(D)\;\ge\;p+q-\gamma^*(a,b), \]
where the \emph{block number} $\gamma^*(a,b)$ is the largest number of blocks into which the entries of the two lines
can be split (Proposition~\ref{prop:twoline}). Figure~\ref{fig:twoline} shows the bound for the example of
Figure~\ref{fig:example}.

\begin{figure}[t]\centering
\begin{tikzpicture}[
  cell/.style={draw, rounded corners=2pt, minimum width=2.2em, minimum height=1.5em, fill=white},
  link/.style={very thick}]
\node[font=\small] at (0,0.8) {row 2};
\node[font=\small] at (3.6,0.8) {row 4};
\node[cell] (a1) at (0,0) {$3$};
\node[cell] (a2) at (0,-0.8) {$5$};
\node[cell] (a3) at (0,-2.1) {$2$};
\node[cell] (b1) at (3.6,-0.4) {$8$};
\node[cell] (b2) at (3.6,-2.1) {$2$};
\draw[link, termB] (a1) -- node[above, font=\small] {$3$} (b1);
\draw[link, termA] (a2) -- node[below, font=\small] {$5$} (b1);
\draw[link, termC] (a3) -- node[below, font=\small] {$2$} (b2);
\begin{scope}[on background layer]
\node[draw=black!45, dashed, rounded corners=4pt, fill=black!4, fit=(a1)(a2)(b1), inner sep=4pt] (B1) {};
\node[draw=black!45, dashed, rounded corners=4pt, fill=black!4, fit=(a3)(b2), inner sep=4pt] (B2) {};
\end{scope}
\node[anchor=west, font=\small] at ([xshift=3pt]B1.east) {block: $3+5=8$};
\node[anchor=west, font=\small] at ([xshift=3pt]B2.east) {block: $2=2$};
\end{tikzpicture}
\caption{The two-line bound for rows~2 and~4 of the matrix in Figure~\ref{fig:example}. The boxes are the nonzero entries
of the two rows. Each edge stands for one term of the decomposition in Figure~\ref{fig:example}, a permutation matrix with
its weight. It joins the two positions that the permutation matrix covers in these rows and carries the weight. The
connected positions form two blocks, shown as dashed frames, and in each block the entries of the two rows have equal
sums. The entries cannot be split into more than two such blocks. Five positions in two blocks need at least $5-2=3$
edges, so every decomposition has at least three terms, also with real weights.}\label{fig:twoline}
\end{figure}
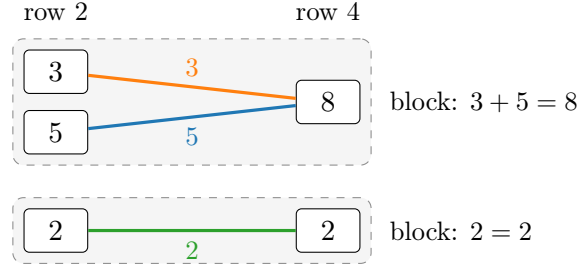

Equal values of the two lines can always form blocks of their own, so $\gamma^*(a,b)=M_2+\gamma^*(a',b')$
(Lemma~\ref{lem:equal}), where $M_2$ is the largest number of disjoint pairs of equal values with one value from each
line, and $a'$ and $b'$ are the multisets of the \emph{unmatched values} of $a$ and $b$. How we compute the block number
of the unmatched values depends on how many there are (Propositions~\ref{prop:blockcount} and~\ref{prop:packing}).

With up to $22$ unmatched values, we compute $\gamma^*(a',b')$ exactly and call the result the \emph{exact block
number}. We give the values of $a'$ the sign $+$ and those of $b'$ the sign $-$. Then $\gamma^*(a',b')$ is the largest
number of nonempty zero-sum prefixes over all orderings of the signed values, which a dynamic program over the $2^m$
subsets of the $m$ values computes. Its table has nine bytes per subset, one for the value $\mathrm{dp}(S)$ and eight
for the sum of $S$, so $m=22$ keeps it at $38$\,MB and the time per pair of lines at about $0.04$\,s. Each further value
doubles both, and the bound chain may visit thousands of pairs. With more unmatched values, we use a \emph{packing
bound}. Every block of unmatched values has at least three values. We list the blocks with three to $T$ values, and all
other blocks have at least $T+1$ values. A linear program over the small blocks then gives an upper bound on
$\gamma^*(a',b')$. Its dual solution is rounded up and checked in integer arithmetic, so rounding can only weaken the
bound. If the packing bound does not raise $L$ to the length of the best decomposition and there are at most $44$
unmatched values, we compute the exact block number in a second way. A meet-in-the-middle method~\cite{horowitz1974}
splits the values into two halves of at most $22$ elements and lists all zero-sum subsets if there are at most two
million of them (otherwise it gives up), and a recursion over the element with the smallest index gives
$\gamma^*(a',b')$.

For all pairs of lines, we first compute a quick upper bound on the block number: a block of three unmatched values
needs a value of one line that equals the sum of two values of the other line, which we find by comparing sums of pairs,
and all other blocks have at least four values. The stronger methods are then applied to the pairs of lines that can
still raise the bound, the most promising pairs first, with at most half of the remaining time. Since they never closed
the gap on dense matrices with $n\ge24$ in our tests, they get at most $2$\,s there.

\subsubsection*{SAT certificates}
For a given $k$, a SAT model decides whether $D$ has a decomposition with precision $1/s$
and $k$ terms. A Boolean variable $y_{l,e}$ states that term $l$ covers position $e$, and each integer weight $c_l$ is
encoded as a binary number. The constraints require that each term covers exactly one position of every line, that every
weight is at least one, that a term covering position $e$ has weight at most $D_e$, that the weights of the terms
covering a position add up to its entry, and that the weights are sorted. The weights then also sum to $s$. The sums are
encoded by binary adders (Section~\ref{sec:implementation}). If the model for $k\le s$ terms is unsatisfiable, then
$k^*(D)>k$ (Proposition~\ref{prop:sat}). Such a certificate concerns only decompositions with the fixed, finite
precision $1/s$ of the data and does not exclude shorter decompositions with arbitrary real weights.

For $n\le7$, we split the question into cases in the style of cube-and-conquer~\cite{heule2011}. For one chosen line, a
case fixes how many terms cover each of its positions and sorts the weights within each of these groups instead of
sorting all weights. A term that covers position $e$ weighs at most $U_e$, the largest value of $\min_iD_{i\pi(i)}$ over
the permutations $\pi$ through $e$, so at least $\lceil D_e/U_e\rceil$ terms cover $e$, and the cases below these
numbers can be skipped. After its terms are sorted by their position in the chosen line and by weight within a position,
every decomposition with $k$ terms belongs to one of the remaining cases. Hence, if all of them are unsatisfiable and
$k\le s$, then $k^*(D)>k$ (Propositions~\ref{prop:sat} and~\ref{prop:cubes}). The cases are independent and are solved
in parallel. Certify calls the SAT model for $k=L,L+1,\dots$ and raises $L$ by one after every unsatisfiable answer. It
does so for matrices treated as sparse with $n\le32$ and for matrices treated as dense with $n\le9$. On larger matrices
treated as dense, the SAT stage never raised a bound in our tests, so Certify stops after the two-line bound.

\subsection{Search: dense matrices}\label{sec:search}

In dense matrices no weight is visible, so the Search threads build decompositions step by step. Their rules are used
for all matrices, but they are the main tool only for dense ones.

\subsubsection*{Zeroing rule}
Let $\pi$ be a permutation whose positions all have positive entries in the residual $R$. The
largest weight with which $P_\pi$ can be subtracted from $R$ without making an entry negative is $\min_iR_{i\pi(i)}>0$,
and we call this subtraction a \emph{zeroing step}, because it makes at least one entry zero. The bottleneck version of
the Birkhoff algorithm makes zeroing steps with the largest possible weight $\theta=\max_\pi\min_iR_{i\pi(i)}$, the
\emph{bottleneck value} of $R$. From now on, the Birkhoff algorithm means this bottleneck version. Among all
permutations $\pi$ with $\min_iR_{i\pi(i)}=\theta$, the \emph{zeroing rule} chooses one that covers the largest number
of positions with entry $\theta$, so that as many entries as possible become zero. This choice requires one assignment
problem per step. The Birkhoff algorithm with the zeroing rule gives the first decomposition of every run.

\subsubsection*{Logarithmic tie-break rule}
The zeroing rule does not determine which permutation to take if several
permutations zero the same number of positions. Among these, the \emph{logarithmic tie-break rule} chooses a permutation
$\pi$ that maximizes
\[ \sum_{e\in\pi:\;R_e>\theta}\log_2\frac{R_e}{R_e-\theta}, \]
where we write $e\in\pi$ for the positions $e$ covered by $\pi$, and the sum runs over those whose entry exceeds
$\theta$. Since $\log_2R_e$ is about the number of binary digits of $R_e$, the rule removes as many digits of the
remaining entries as possible, while each step remains a zeroing step with weight $\theta$ that requires one assignment
problem. On dense matrices with $n\ge32$, one deterministic run of the Birkhoff algorithm with this rule gives
decompositions that are $3$ to $6\,\%$ shorter than with the zeroing rule alone.

Not every decomposition is a sequence of zeroing steps. A decomposition has a \emph{zeroing order} if its indices can be
ordered as $l_1,\dots,l_k$ such that, for every $t$, some position covered by $\pi_{l_t}$ is covered by none of the
later permutations $\pi_{l_{t'}}$, $t'>t$. A decomposition with a zeroing order is a sequence of zeroing steps
(Proposition~\ref{prop:zeroing}). Conversely, in every sequence of zeroing steps, the entry that step $t$ makes zero is
covered by no later term, so the order of the steps is a zeroing order.\footnote{For the ten dense test instances with
$n=5$, the library provides decompositions with precision $1/s$ computed by a mixed-integer solver, six of them proven
optimal. Eight of the ten have a zeroing order, and for the other two, the beam search finds different decompositions of
the same length that have one. Some of the best known decompositions for $n=6$ have no zeroing order.}

\subsubsection*{Beam search}
Beam search~\cite{lowerre1976,ow1988} is a breadth-first tree search that keeps only the $W$
best nodes of each level, where $W$ is the \emph{width}. In our search tree, the root is $D$, a node at depth $t$ is a
residual after $t$ zeroing steps, and a \emph{successor} of a node is the residual after one more zeroing step. For
$n\le7$, we use the zeroing steps with all permutations in the nonzero pattern of the residual, taken from the list of
the at most $7!=5040$ permutations in the nonzero pattern of $D$. By Proposition~\ref{prop:zeroing}, every decomposition
with a zeroing order is then a path in this tree. For larger $n$, there is one successor for each position $e$ with
$0<R_e\le\theta$, which uses a permutation that covers $e$, covers only positions with entries of at least $R_e$, and
covers the largest number of positions with entry $R_e$. A position without such a permutation gives no successor, but
the positions with entry $\theta$ on the bottleneck permutation always give one.

Figure~\ref{fig:beam} shows the first two levels of a small search tree. Successors with the same residual are merged,
which we detect by a hash value of the residual (Section~\ref{sec:implementation}). A successor is pruned if its depth
plus the degree of its residual is at least the length of the best known decomposition, because the residual needs at
least that many further terms (Lemma~\ref{lem:partition}). The other successors are scored by their depth plus the
number of steps that the Birkhoff algorithm with the zeroing rule needs to decompose their residual completely. We call
these steps the \emph{greedy completion} of the successor. The $W$ successors with the lowest scores form the next
level. The steps on the path to a successor, followed by its greedy completion, form a full decomposition. When all
greedy completions of a level are computed, every decomposition that improves the best one is offered to the shared best
decomposition. When the tree is exhausted, the search starts again with three times the width, until the time limit.

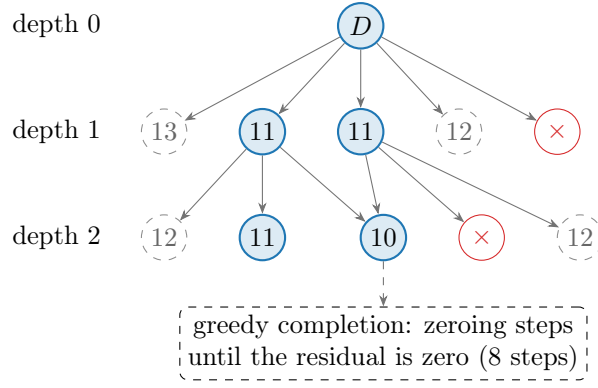
\begin{figure}[t]\centering
\begin{tikzpicture}[
  st/.style={draw, circle, minimum size=6mm, inner sep=0pt, font=\small},
  kept/.style={st, fill=lbcolor!15, draw=lbcolor, thick},
  drop/.style={st, draw=black!45, text=black!60, dashed},
  cut/.style={st, draw=kcolor, text=kcolor},
  ed/.style={-{Stealth[length=1.5mm]}, black!55}]
\node[font=\small, anchor=east] at (-3.3,0) {depth $0$};
\node[font=\small, anchor=east] at (-3.3,-1.4) {depth $1$};
\node[font=\small, anchor=east] at (-3.3,-2.8) {depth $2$};
\node[kept] (r) at (0,0) {$D$};
\node[drop] (a) at (-2.6,-1.4) {$13$};
\node[kept] (b) at (-1.3,-1.4) {$11$};
\node[kept] (c) at (0,-1.4) {$11$};
\node[drop] (d) at (1.3,-1.4) {$12$};
\node[cut] (e) at (2.6,-1.4) {$\times$};
\foreach \x in {a,b,c,d,e} \draw[ed] (r) -- (\x);
\node[drop] (f) at (-2.6,-2.8) {$12$};
\node[kept] (g) at (-1.3,-2.8) {$11$};
\node[kept] (h) at (0.3,-2.8) {$10$};
\node[cut] (i) at (1.6,-2.8) {$\times$};
\node[drop] (j) at (2.9,-2.8) {$12$};
\draw[ed] (b) -- (f); \draw[ed] (b) -- (g); \draw[ed] (b) -- (h); \draw[ed] (c) -- (h); \draw[ed] (c) -- (i);
\draw[ed] (c) -- (j);
\node[draw, dashed, rounded corners=3pt, font=\small, align=center, inner sep=3pt] (comp) at (0.3,-4.2)
  {greedy completion: zeroing steps\\ until the residual is zero (8 steps)};
\draw[ed, dashed] (h) -- (comp);
\end{tikzpicture}
\caption{The first two levels of the beam search with width $2$, so two nodes are kept per level. Each node is a
residual, the part of the matrix that is not yet decomposed after the steps on its path. The number in a node is its
score, the steps so far plus the steps that the Birkhoff algorithm with the zeroing rule still needs to decompose the
residual. Nodes with a thick border have the lowest scores and are kept, and nodes with a dashed border are dropped.
Nodes marked $\times$ are pruned, because the steps so far plus the largest number of nonzero entries in a row or column
of the residual are at least the length of the best decomposition known. Both kept nodes of depth~1 lead to the same
residual with score $10$, which is kept only once. Its greedy completion has $8$ steps, which together with the $2$
steps on its path give a decomposition of length $10$.}\label{fig:beam}
\end{figure}

\subsubsection*{Improvement}
Further threads try to improve the best decomposition. The \emph{SAT re-optimization of the
tail} takes the \emph{tail} of the best decomposition, that is, the longest list of its lightest terms whose sum has no
entry above a cap, and asks the SAT model of Section~\ref{sec:certify} for a decomposition of this sum with fewer terms.
The caps $15$, $31$ and $63$ are tried in this order. With them, every entry of the sum fits into $4$, $5$ or $6$ bits,
which keeps the model small. A large neighborhood search~\cite{shaw1998,pisinger2010} removes a group of terms, namely a
random term with a small weight together with the terms that share the most positions with it, and decomposes the sum of
these terms again with the \emph{randomized zeroing rule}. This rule adds to the assignment cost of every position a
random number below half the value of one zeroed position and thus chooses at random among nearly equal permutations.
Finally, \emph{randomized restarts} run the Birkhoff algorithm with a randomized rule on the whole matrix $D$: on
matrices treated as dense with $n\ge12$, they use the logarithmic tie-break rule with random changes of the assignment
costs by less than $0.3$ binary digits, and otherwise the randomized zeroing rule.

\subsection{Parallel run}\label{sec:run}

Table~\ref{tab:threads} shows how a run uses its six threads. The threads of the large neighborhood search also start a
randomized restart in every fourth round. For $n\le7$, the beam search over all permutations finishes its first passes
within seconds, so the cases of the SAT certificates take over one or two of its threads. For $n\le4$, one thread calls
the SAT model for $k=L,L+1,\dots$ and keeps the first decomposition it finds. On matrices treated as dense, this is a
seventh thread.

\begin{table}[t]\centering\small
\caption{Threads of one run. A matrix is treated as sparse if at least $45\,\%$ of its nonzero entries have a value that
occurs more than once, and as dense otherwise. Certify forms the bound chain, and all other threads form the solution
chain. For small sizes, the cases of the SAT certificates take over beam search threads, and one thread searches for a
decomposition with the SAT model.}\label{tab:threads}
\begin{tabular}{@{}l>{\raggedright\arraybackslash}p{3.9cm}>{\raggedright\arraybackslash}p{5.6cm}@{}}\toprule
thread & treated as sparse & treated as dense \\ \midrule
1 & Certify & Certify, then randomized restarts ($n\ge24$) or large neighborhood search ($10\le n<24$) \\
2 & fixing by unique subset sums, then enumeration of weight sets & beam search \\
3 & largest weight first & beam search, cases of Certify for $n\le5$ \\
4 & large neighborhood search & beam search, cases of Certify for $n\le7$ \\
5 & large neighborhood search & SAT re-optimization of the tail \\
6 & large neighborhood search, SAT model for $n\le4$ & large neighborhood search \\
7 & & SAT model for $n\le4$ \\
\bottomrule\end{tabular}
\end{table}

\section{Implementation}\label{sec:implementation}

RCS is implemented in Rust. The code was written with the assistance of AI tools (see the statement on the use of
generative AI), and every reported decomposition is verified with the independent checker of the library
(Section~\ref{sec:setup}). This section describes the data structures and encodings on which the parts of
Section~\ref{sec:method} rely. All computations on the matrix use exact integer arithmetic. The matrix $D$ is stored row
by row as 64-bit integers, and the instance reader reads only the matrix, its size and its scale, never the generating
permutations and weights that the instance files also contain.

\subsubsection*{Matching}
Each step of the Birkhoff algorithm needs the bottleneck value $\theta$ and an assignment problem
on the \emph{threshold set} $\{e:R_e\ge\theta\}$, the set of positions with entry at least $\theta$. We keep the
positions with positive entries as keys sorted by value. After a step, only $n$ positions change and are merged into the
sorted list, so the threshold set of every threshold is a suffix of this list. Since $\theta$ never increases during a
run of the Birkhoff algorithm, the search for $\theta$ starts at the previous value and moves down in doubling steps.
Every test of a threshold is a perfect matching problem on its suffix, which we solve by the algorithm of Hopcroft and
Karp~\cite{hopcroft1973} on a compressed adjacency list, starting from the previous matching. The assignment problem is
solved by successive shortest paths with a binary heap on the threshold set only. Among columns at the same distance, we
take the one with the smallest index, as the Hungarian method~\cite{kuhn1955} does. With this tie-break, the Birkhoff
algorithm with the zeroing rule gives the same decompositions as with the Hungarian method on all test instances and is
$2$ to $4$ times faster for $n=12$ to $100$. All work arrays are allocated once and reused.

\subsubsection*{Packing bound}
The zero-sum subsets with three to $T$ elements, for $T=4,5,6$, are found by merging sorted
tables of subset sums of the two lines, and the packing linear program is solved by a dense simplex method with Bland's
rule. Its dual values are rounded up to multiples of $1/(1000(T+1))$, and every block constraint is checked in integer
arithmetic and repaired if necessary before the floor is taken.

\subsubsection*{SAT encodings}
The weights are encoded as binary numbers, and the sums use adder trees whose partial sums
are bounded by the target entry, so every adder has only $\lceil\log_2(D_e+1)\rceil$ bits and its overflow is fixed to
zero. A sequential counter~\cite{sinz2005} encodes that a term covers exactly one position per line, and a comparator
with a constant encodes the bound $c_l\le D_e$ for a term that covers position $e$. These are the usual translations of
pseudo-Boolean constraints into SAT~\cite{een2006}. A deadline and a stop flag end the solver when the run ends or, for
the cases of Section~\ref{sec:certify}, when one case is satisfiable. The bounds $U_e$ are computed as
$\min(D_e,\theta_e)$, where $\theta_e$ is the bottleneck value of $D$ without the row and the column of $e$. The SAT
re-optimization of the tail skips a cap if the number of tail terms times the number of positions of their sum exceeds
$50\,000$, which keeps the model below about one gigabyte.

\subsubsection*{Beam search}
For $n\le7$, the initial width is $200\,000$ divided by the number of permutations in the
nonzero pattern of $D$, and for $8\le n\le12$, $13\le n\le16$, $17\le n\le32$ and $n>32$, it is $30$, $10$, $3$ and $1$.
The successors of a level are never stored all at once. A first parallel pass computes a 64-bit hash value and the
degree of the residual for every pair of a node and a zeroing step, and a parallel sort merges successors with equal
hash values. A second pass rebuilds the remaining residuals and computes the lengths of their greedy completions. A
completion is stored as a decomposition only if it improves the best decomposition, and only the $W$ best successors are
kept. If there are at most $120$ candidate permutations, as for all $n\le5$, the greedy completions scan this list
instead of solving matching problems. For the successor that belongs to a position $e$ whose value occurs at no other
position, every perfect matching through $e$ inside $\{R\ge R_e\}$ zeroes exactly one position, so it is found by
rerouting the bottleneck matching of the node through $e$ along one augmenting path. For the other positions, we solve
an assignment problem without the row and the column of $e$.

\subsubsection*{Threads and memory}
Every improvement is written to disk at once, together with its time, via a temporary
file and a rename, so an interrupted run loses nothing, and the time series of the best decomposition is available for
every run. The beam search uses its own thread pool, and the memory allocator is mimalloc~\cite{mimalloc}.

\section{Computational experiments}\label{sec:results}

\subsection{Experimental setup}\label{sec:setup}

\subsubsection*{Test instances}
We use the \NInst{} test instances of the minimum Birkhoff decomposition problem from the
Quantum Optimization Benchmarking Library~\cite{koch2025}, \NSparse{} sparse and \NDense{} dense matrices: ten instances
for each class and each size $n=3,\dots,16$ and $n=24$, $32$, $48$, $64$, $100$, except that there are only five dense
instances for $n=100$. The sparse instances have $n$ generating permutations and the dense instances $n^2$. The entries
have $3$ to $7$ decimal digits, so $s$ ranges from $10^3$ to $10^7$ and grows with $n$, and the generating weights have
the same precision. Appendix~\ref{app:instances} describes the instances in more detail.

\subsubsection*{Hardware and software}
All computations were run on an Apple MacBook Pro with an M4 Pro processor (12 cores)
and 24\,GB of memory under macOS 27.0.1. RCS was compiled with Rust 1.89.0 and uses the SAT solver CaDiCaL
1.9.5~\cite{cadical2020} through its Rust interface and the library rayon 1.12.0 for parallel loops.

\subsubsection*{Protocol}
Every run used six threads,\footnote{Seven for matrices with $n\le4$ that RCS treats as dense
(Table~\ref{tab:threads}).} and two instances were run at the same time. A run consists of two rounds. In the first
round, each instance has a time limit of 60\,s, and the instances without a proof of optimality get a second round with
a time limit of 600\,s, which starts from scratch. The result of a run is the shorter decomposition of its two rounds.
RCS is randomized: several threads draw random numbers, and since the threads share the best decomposition, the result
also depends on their timing. We made \NRuns{} independent runs with the seeds $0,\dots,4$ of the random number
generators. For each instance, we report the shortest decomposition over the runs (on a tie, the run that proved
optimality itself, then the faster run) and the largest lower bound over all rounds. All results come from these runs of
one version of the algorithm, and every lower bound and every proof of optimality comes from the runs themselves. All
decompositions were also verified with the independent checker that comes with the instances.

\subsubsection*{Measures}
For each instance, $k$ denotes the length of the best decomposition over the runs and $L$ the
lower bound. The time to the best decomposition of a run is the time from its start until it first found a decomposition
of its final length, where the second round counts from the end of the first. The running time of a run is the time of
both rounds. For each instance, both times are means over the \NRuns{} runs. Times over groups of instances are
aggregated by the shifted geometric mean with a shift of 1\,s. Appendix~\ref{app:results} gives the results per family
of instances, and Appendix~\ref{app:all} gives them for every instance.

\subsection{Solution quality and optimality}\label{sec:quality}

Table~\ref{tab:groups} summarizes the results by class and size, and Figure~\ref{fig:bounds} shows $k$ and $L$ for every
instance. RCS proves optimality for \NProven{} of the \NInst{} instances, namely for \NProvenSparse{} of the \NSparse{}
sparse and \NProvenDense{} of the \NDense{} dense ones. For \NProvenReal{} of them, the proof comes from the degree or
the two-line bound, so the decomposition is optimal even among all decompositions with real weights, $k=k^*_{\R}(D)$.
For the other \NProvenSatOnly{} instances, the proof is a SAT certificate and holds only for precision $1/s$.

All \NRuns{} runs give the same length on \NRunsSame{} instances, namely on all \NSparse{} sparse instances, on all
\NProvenDense{} dense instances with a proof of optimality and on \NRunsSameDenseOpen{} further dense instances. On the
other \NRunsDiffer{} instances, which are all dense, the lengths of the runs differ by at most \MaxRunSpread{}, and on
\NRunsSingleBest{} of them only one run reaches the best length. For $n\ge24$, the standard deviation of the length over
the runs is \DenseLargeRunStd{} on average.

\begin{table}[t]\centering\small
\caption{Results by class and size. Sparse matrices are built from $n$ permutation matrices and dense ones from $n^2$.
The columns give the number of instances, the mean over the instances of the length of the best decomposition divided by
the size $n$ and the same mean for the lower bound, the number of instances whose decomposition is proven optimal, the
time to the best decomposition and the running time, both means over the runs. Times over instances are shifted
geometric means with a shift of 1\,s. A run stops early when its decomposition is proven optimal.}\label{tab:groups}
\GroupTable
\end{table}

\begin{figure}[t]\centering
\pgfplotsset{/pgfplots/bounds/.style={width=\linewidth, height=6cm, xmin=0, xmax=\FigBoundsXmax,
  xtick/.expanded={\FigBoundsTicks}, xticklabels/.expanded={\FigBoundsTickLabels},
  extra x ticks/.expanded={\FigBoundsTicksLow}, extra x tick labels/.expanded={\FigBoundsTickLabelsLow},
  extra x tick style={xticklabel style={yshift=-1.6ex}}, x tick label style={font=\scriptsize},
  every x tick/.style={draw=none}, xlabel={size $n$}, ymajorgrids}}
\begin{tikzpicture}
\begin{axis}[bounds, ymode=log, ymin=1, ymax=1000, ytick={1,2,5,10,20,50,100,200,500},
  yticklabels={1,2,5,10,20,50,100,200,500}, ylabel={length}, legend style={at={(0.02,0.98)}, anchor=north west}]
\addplot[forget plot, draw=none, mark=none, error bars/.cd, y dir=minus, y explicit, error mark=none,
  error bar style={kcolor!35, line width=0.3pt}] table[x=x, y=k, y error expr=\thisrow{k}-\thisrow{l}]
  {fig_bounds_open.csv};
\addplot[only marks, mark=triangle*, mark size=0.75pt, lbcolor, legend image post style={mark size=1.6pt}]
  table[x=x, y=l] {fig_bounds_open.csv};
\addlegendentry{lower bound, unproven instance}
\addplot[only marks, mark=*, mark size=0.4pt, kcolor, legend image post style={mark size=1.2pt}]
  table[x=x, y=k] {fig_bounds_open.csv};
\addlegendentry{length, unproven instance}
\addplot[only marks, mark=diamond*, mark size=0.85pt, optcolor, legend image post style={mark size=1.8pt}]
  table[x=x, y=k] {fig_bounds_proven.csv};
\addlegendentry{proven optimal}
\draw[black, line width=0.4pt] (axis cs:\FigBoundsSeparator,1) -- (axis cs:\FigBoundsSeparator,1000);
\node[font=\small, anchor=south] at (axis cs:\FigBoundsMidSparse,1.12) {sparse};
\node[font=\small, anchor=south] at (axis cs:\FigBoundsMidDense,1.12) {dense};
\end{axis}
\end{tikzpicture}
\caption{Length of the best decomposition and lower bound for every test instance, on a logarithmic scale. Diamonds mark instances for which the two are
equal, so that the decomposition is proven optimal. For every other instance, a triangle marks the lower bound, a circle
marks the length, and a vertical line connects them. The black vertical line separates the sparse instances (left,
built from $n$ permutation matrices) from the dense instances (right, built from $n^2$ permutation matrices). Within
each class, the instances are grouped by size, and the number below each group is its size $n$.}\label{fig:bounds}
\end{figure}
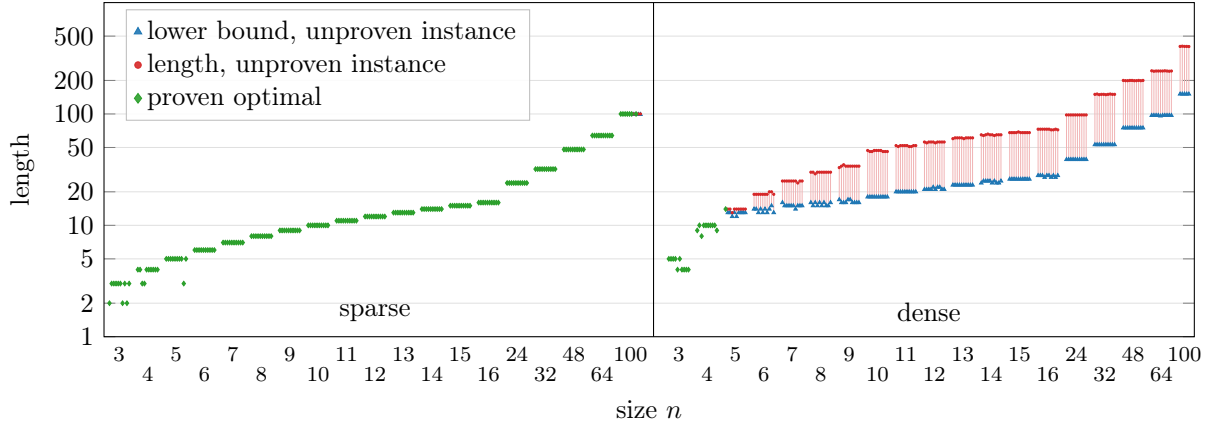

RCS finds a decomposition of length $n$ for every sparse instance with $n\ge6$ and a shorter one for some sparse
instances with smaller $n$. The median time to the best decomposition is \SparseMedianTts\,s, and the longest in any run is
\SparseMaxTts\,s. Only \NOpenSparse{} sparse instances remain without a proof of optimality, each with $k=L+1$.

The dense decompositions are much longer. Their mean length divided by $n$ is \DenseRatioFive{} for $n=5$ and between
\DenseRatioMidMin{} and \DenseRatioMidMax{} for $n\ge12$, so RCS needs only a small fraction of the $n^2$ generating
permutations. This ratio jumps where the scale $s$ grows by a factor of ten, from $n=9$ to $n=10$, from $n=24$ to $n=32$
and from $n=64$ to $n=100$, and it decreases between these sizes (Figure~\ref{fig:bounds} and
Appendix~\ref{app:instances}). The number of decimal digits of the data therefore has a strong effect on the length.

\subsection{Lower bounds}\label{sec:bounds}

Table~\ref{tab:stages} shows which lower bound proves optimality. For most of the proven instances, it is the two-line
bound. In sparse matrices, most entries are single weights, so two lines share most of their values. After equal values
have been paired, few unmatched values remain, and the exact block number or the packing bound shows that they cannot
form many blocks. The packing bound is needed for the large instances, because for $n\ge32$ the promising pairs of lines
have $34$ to $80$ unmatched values, too many for the dynamic program over subsets.

In dense matrices, two lines share few values, and the two-line bound stays between $1.5\,n$ and $1.7\,n$ for $n\ge24$.
The SAT certificates prove optimality for all dense instances with $n=4$ and for one instance with $n=5$. They raise the
lower bound above the two-line bound for all dense instances with $5\le n\le7$, for seven with $n=8$ and for three with
$n=9$, and the mean lower bound decreases from $2.5\,n$ for $n=5$ to $1.8\,n$ for $n=9$. These bounds hold only for
precision $1/s$. For $n\ge24$, the mean lower bound is $\DenseLbRatioLarge\,n$, while the mean length per size lies
between $\DenseRatioBigMin\,n$ and $\DenseRatioBigMax\,n$.

\begin{table}[t]\centering\small
\caption{Number of instances by the lower bound that proves optimality. The degree is the largest number of nonzero
entries in a row or a column. The two-line bound uses the entries of two rows or two columns. A SAT certificate shows
with a satisfiability solver that no decomposition with fewer terms and with weights that are multiples of $1/s$ exists.
The column ``valid for'' says whether the proof holds for all real weights or only for weights that are multiples of
$1/s$, like the data.}\label{tab:stages}
\StageTable
\end{table}

\subsection{Running times}\label{sec:runtime}

Figure~\ref{fig:tts} shows the time to the best decomposition for every instance. On sparse matrices, every run finds the
best decomposition after at most \SparseMaxTts\,s for all sizes up to $n=100$, and \SparseProvenFast{} of the
\NProvenSparse{} proofs of optimality take less than one second in every run. On dense matrices, the first decomposition is available
after less than $0.2$\,s for all sizes, and most of the improvement happens in the first seconds
(Figure~\ref{fig:progress}). In the run that found the best decomposition, for $6\le n\le16$, the median time to reach
length $k+2$ is \ReachTwoMid\,s, and the median time to reach $k+1$ is \ReachOneMid\,s. At $n=24$, these times are
\ReachTwoTwentyFour\,s and \ReachOneTwentyFour\,s, and at $n=32$ they are \ReachTwoThirtyTwo\,s and
\ReachOneThirtyTwo\,s. The last one or two terms take most of the time, often tens to hundreds of
seconds. From $n=24$ on, the beam search is slow, because one greedy completion takes $O(kn^3)$ time, the first level
has thousands of successors, and the improvements of a level are offered only when the level is complete. For these
sizes, the large neighborhood search and the randomized restarts find the best decompositions
(Section~\ref{sec:components}).

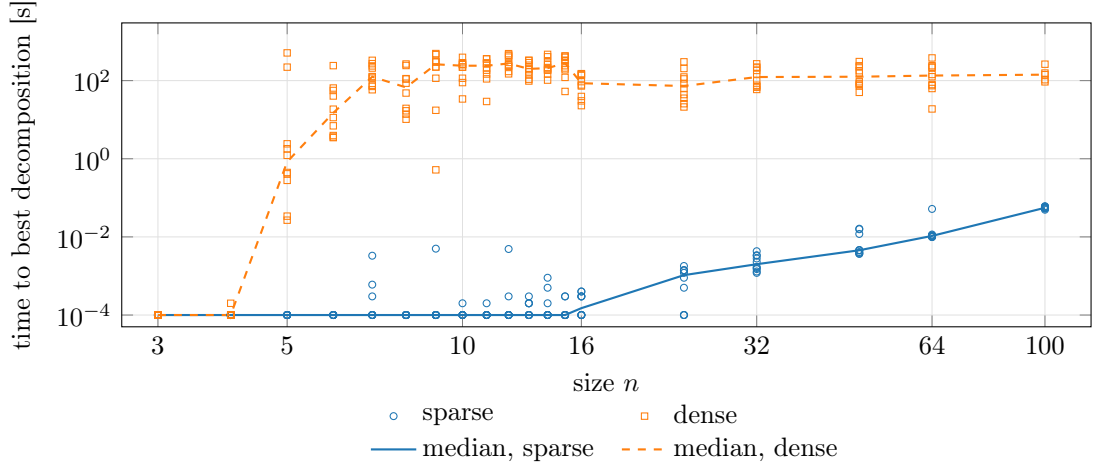
\begin{figure}[t]\centering
\begin{tikzpicture}
\begin{axis}[width=0.9\linewidth, height=5.6cm, xmode=log, ymode=log, log basis x=10, xmin=2.6, xmax=120,
  ymin=5e-5, ymax=3000, xtick={3,5,10,16,32,64,100}, xticklabels={3,5,10,16,32,64,100}, xlabel={size $n$},
  ylabel={time to best decomposition [s]}, grid=major, legend columns=2, legend style={draw=none, at={(0.5,-0.22)},
  anchor=north, /tikz/every even column/.append style={column sep=8pt}}]
\addplot[only marks, mark=o, mark size=1.3pt, sparsecolor] table[x=n, y=tts] {fig_tts_sparse.csv};
\addlegendentry{sparse}
\addplot[only marks, mark=square, mark size=1.2pt, densecolor] table[x=n, y=tts] {fig_tts_dense.csv};
\addlegendentry{dense}
\addplot[thick, sparsecolor] table[x=n, y=tts] {fig_tts_median_sparse.csv};
\addlegendentry{median, sparse}
\addplot[thick, densecolor, dashed] table[x=n, y=tts] {fig_tts_median_dense.csv};
\addlegendentry{median, dense}
\end{axis}
\end{tikzpicture}
\caption{Time to the best decomposition against the size $n$ for every instance, as the mean over the runs, both on
logarithmic scales. Circles
mark sparse instances and squares dense instances. The solid line connects the medians per size for the sparse
instances, and the dashed line those for the dense instances. Times below $10^{-4}$\,s are shown at
$10^{-4}$\,s.}\label{fig:tts}
\end{figure}

\begin{figure}[t]\centering
\begin{tikzpicture}
\begin{axis}[width=0.9\linewidth, height=5.6cm, xmode=log, log basis x=10, xmin=1e-3, xmax=700, ymin=0, ymax=17.5,
  xlabel={time [s]}, ylabel={excess over the final length [\%]}, const plot, grid=major,
  legend style={at={(0.98,0.98)}, anchor=north east}]
\addplot[thick, tabblue] table[x=t, y=n8] {fig_progress.csv};
\addlegendentry{$n=8$}
\addplot[thick, taborange, densely dashed] table[x=t, y=n16] {fig_progress.csv};
\addlegendentry{$n=16$}
\addplot[thick, tabgreen, dotted] table[x=t, y=n32] {fig_progress.csv};
\addlegendentry{$n=32$}
\addplot[thick, tabred, dashdotted] table[x=t, y=n64] {fig_progress.csv};
\addlegendentry{$n=64$}
\addplot[thick, tabpurple, loosely dashed] table[x=t, y=n100] {fig_progress.csv};
\addlegendentry{$n=100$}
\end{axis}
\end{tikzpicture}
\caption{Progress on dense instances of five sizes $n$, in the run that found the best decomposition. Each curve shows
by how many percent the best decomposition found so far is longer than the final decomposition, as the median over the ten instances of this size
(five for $n=100$). Each curve starts when the first decomposition is available.}\label{fig:progress}
\end{figure}
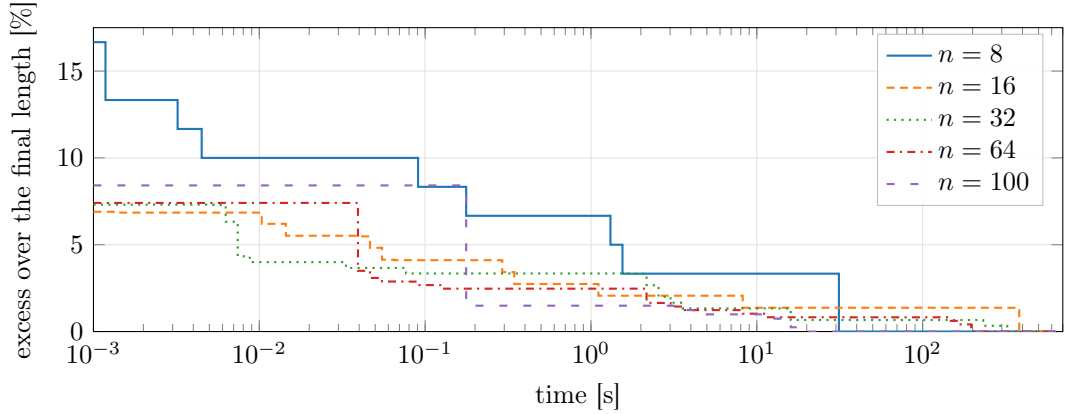

\subsection{Contribution of the components}\label{sec:components}

Table~\ref{tab:methods} shows which component found the best decomposition. On many small sparse instances, the Birkhoff
algorithm with the zeroing rule is already optimal, and on almost all others Recover gives the result, through fixing by
unique subset sums or the enumeration of weight sets. One dense instance with $n=3$ is treated as
sparse (Section~\ref{sec:problem}), so Recover runs on it, and the enumeration of weight sets finds its best
decomposition. On dense instances with $n\le16$, the beam search finds most of the best decompositions, whereas on
larger ones it never gives the best decomposition. There, the large neighborhood search, the randomized restarts and the
SAT re-optimization of the tail make the last improvements.

\begin{table}[t]\centering\small
\caption{Number of instances for which each component of RCS found the best decomposition, in the run that is reported.
The first component gives the first decomposition of every run. The others either build their own decompositions or improve the best
one.}\label{tab:methods}
\MethodTable
\end{table}

\subsection{Comparison with other methods}\label{sec:comparison}

Other groups have submitted results for these test instances to the library~\cite{qoblibsub}. We checked every submitted
decomposition in integer arithmetic and counted it only if it reproduces $D$ exactly. Table~\ref{tab:comparison}
compares RCS with the five submissions that contain exact decompositions: Birkhoff+~\cite{valls2021}, a mixed-integer
program solved with Gurobi (MIP), an extended fully corrective Frank--Wolfe method (E-FCFW) whose candidate permutations
are sampled by the quantum approximate optimization algorithm (QAOA), a quantum black-box optimization method, and a
partitioned variant of QAOA with E-FCFW. Since these submissions ran on other computers and with other time limits,
their running times can be compared only roughly.

\begin{table}[t]\centering\small
\caption{Comparison with the exact decompositions of other methods. For each method, the table gives the range of sizes
$n$ and the number of instances for which the method gives an exact decomposition. On these instances, it counts how
often the decomposition of RCS is shorter than, as long as, or longer than that of the method. The ratio is the mean of
the length of the method divided by the length of RCS. The last column gives the number of these instances proven
optimal by the method and by RCS.}\label{tab:comparison}
\CompTable
\end{table}

No submitted exact decomposition is shorter than the decomposition of RCS. Birkhoff+ covers \CompBplusN{} instances with
$n$ up to $100$ and is fast, with a median reported running time of \CompBplusTime\,s per instance, but its
decompositions are on average \CompBplusRatio{} times as long. On the same instances, the median time at which RCS finds
its best decomposition is \CompBplusOurTime\,s. The MIP covers the \CompMipN{} instances with $n\le6$ and proves
optimality for \CompMipProven{} of them, whereas RCS proves optimality for \CompMipOurProven{}. The difference consists
of five dense instances with $n=5$, for which RCS finds decompositions of the same length, but its lower bound stays one
below. On the ten dense instances with $n=6$, the MIP stops at its time limit of two hours, and the decompositions of
RCS are shorter by one to three terms. E-FCFW gives exact decompositions for \CompEfcfwN{} of its $131$ instances, all
with $n\le8$, and approximate ones for the others, with a largest entry error between $5.5\cdot10^{-5}$ and
$2.7\cdot10^{-2}$. On the \CompEfcfwN{} instances with exact decompositions, RCS finds decompositions of the same length
for \CompEfcfwEqual{} and shorter ones for \CompEfcfwShorter, and it proves all of them optimal. The quantum black-box
method gives exact decompositions for \CompQbN{} of its $40$ instances with $n=4,5$. RCS is shorter on \CompQbShorter{}
of them and as long on the others. The partitioned variant of QAOA gives one exact decomposition, of length $66$ for an
instance with $n=16$, for which RCS finds $16$ terms.

Of the other submissions, one contains only approximate decompositions, and three contain no solution files, so their
reported lengths cannot be checked. None of the four appears in the table. For Blended Frank--Wolfe and a first version
of Birkhoff+, the reported lengths are never shorter than those of RCS. For the integer program of Koch et
al.~\cite{koch2025}, solved with CPLEX ($n\le8$), the reported lengths of $20$ dense instances lie below proven lower
bounds and therefore cannot be lengths of exact decompositions with weights of precision $1/s$, which an integer program
produces. For $n=7$, the reported length is $7$, but a decomposition of length $7$ would cover each position of a row
with seven nonzero entries by exactly one term, so all such rows would consist of the same seven weights, which is not
the case in any of the ten instances. The two-line bound gives $13$ for these instances. For $n=5$, the reported lengths
of $8$ to $10$ lie below the lower bounds of $12$ to $14$ that both the MIP and our SAT certificates prove for precision
$1/s$.

\section{Conclusion and future work}\label{sec:conclusions}

We presented RCS, a method for exact minimum Birkhoff decompositions of doubly stochastic matrices with finite
precision, and tested it on the \NInst{} test instances of the Quantum Optimization Benchmarking Library.

For sparse matrices, the problem is solved in practice: the degree and the two-line bound, which hold for real weights,
prove optimality for \NProvenSparseReal{} of the \NSparse{} sparse instances, and the \NOpenSparse{} open ones have
$k=L+1$. Dense matrices are much harder: the beam search over zeroing steps and, for large $n$, the large neighborhood
search and the randomized restarts give short decompositions, but optimality is proven only up to $n=5$, and for $n\ge4$
only for weights with precision $1/s$. Compared with the exact decompositions known before, RCS is shorter on \NBetter{}
instances, equal on \NEqual{} and never longer, and it gives the first exact decompositions for \NNew{} instances.

Not every idea worked: sampling vertices of the linear relaxation over permutation matrices gave $16$ to $17$ terms on
dense instances with $n=5$, where decompositions with $13$ to $14$ terms exist, and plain linear programming bounds were
always weaker than the two-line bound.

The main open question is the gap for large dense matrices, where the decompositions are about \DenseGapFactor{} times
as long as the lower bound (Section~\ref{sec:bounds}), and we do not know whether the bounds or the decompositions are
farther from the optimum. The two-line bound alone cannot close this gap, because two lines of a dense matrix allow many
blocks, so stronger bounds will need arguments that involve more lines. It also remains open whether real weights can
give shorter decompositions for matrices of this kind, since our SAT certificates hold only for the fixed precision
$1/s$. Finally, the running time for large dense matrices is a practical problem, because the beam search spends most of
it on greedy completions.

\section*{Code and data availability}

The source code of RCS (Rust), the scripts that run the experiments and the scripts that
produce all tables and figures of this paper are available at \url{https://github.com/tg90000/rcs-birkhoff}.

The test instances are available from the Quantum Optimization Benchmarking
Library~\cite{koch2025} at \url{https://github.com/ZIB-AOPT/QOBLIB}. All decompositions, lower bounds and run logs
generated in this study are available at \url{https://github.com/tg90000/rcs-birkhoff}.

\section*{Use of generative AI}

AI assisted in coding and in formulating the manuscript. Claude Code (Anthropic) with
the model Claude Opus 5.5 assisted in writing the Rust code, the experiment scripts and the text, and Claude Fable 5.1
assisted in challenging ideas and in verifying them. The author directed the work, checked the code, the proofs and the
results, and is responsible for the content.

\section*{Acknowledgments}

We thank Dr.~Corey O'Meara for conversations on the problem and on hybrid solutions.

\bibliographystyle{plainurl}
\bibliography{references}

\clearpage
\appendix

\section{Proofs}\label{app:proofs}

The proofs follow the order in which Sections~\ref{sec:problem} and~\ref{sec:method} use the results.

\begin{proof}[Proof of Lemma~\ref{lem:partition}]
Comparing the entries at position $(i,j)$ on both sides of $D=\sum_lc_lP_{\pi_l}$ gives $D_{ij}=\sum_{l\in I_{ij}}c_l$.
Since all $c_l$ are positive, $I_{ij}$ is empty if and only if $D_{ij}=0$. For a fixed row $i$, every permutation
$\pi_l$ maps $i$ to exactly one column, so every $l$ lies in exactly one of the sets $I_{ij}$, $j=1,\dots,n$. Dropping
the empty ones leaves the sets with $D_{ij}>0$, which therefore partition $\{1,\dots,k\}$. For a fixed column $j$, every
$\pi_l$ maps exactly one row to $j$, and the same argument applies to the sets $I_{ij}$, $i=1,\dots,n$. Summing $D_{ij}$
over the columns $j$ of any row gives $\sum_lc_l=s$, because the sets $I_{ij}$ of the row partition $\{1,\dots,k\}$.
Since the partition of a line has one nonempty set for each nonzero entry of the line, $k\ge\deg(D)$ for every
decomposition with real weights, so $k^*_{\R}(D)\ge\deg(D)$, and $k^*(D)\ge k^*_{\R}(D)$ because every decomposition
with precision $1/s$ is a decomposition with real weights.
\end{proof}

\begin{theorem}[soundness of fixing by unique subset sums]\label{thm:sweep}
Let $(c_l,\pi_l)_{l=1}^k$ be a decomposition of $D$ whose weights, indexed so that $c_1\le\dots\le c_k$, are the weights
given to fixing by unique subset sums, and let $I^*_{ij}=\{l:\pi_l(i)=j\}$. Then every position $(i,j)$ that the passes
of the procedure fix is fixed to $I^*_{ij}$. In particular, all decompositions with these weights agree at the fixed
positions.
\end{theorem}
\begin{proof}
We proceed by induction over the fixing steps. Initially, $F_i\cap G_j=\{1,\dots,k\}$ for all $i$ and $j$, so every set
$I^*_{ij}$ is contained in $F_i\cap G_j$. Assume that all positions fixed so far are fixed to their sets $I^*$. Consider
an open position $(i,j)$. The indices removed from $F_i$ are exactly those of the sets $I^*_{ij'}$ of the fixed
positions $(i,j')$ of row $i$. These have $j'\ne j$, so they are disjoint from $I^*_{ij}$ by Lemma~\ref{lem:partition},
and the same holds for $G_j$. Hence $I^*_{ij}\subseteq F_i\cap G_j$, and $I^*_{ij}$ has sum $D_{ij}$ by
Lemma~\ref{lem:partition}. A position is fixed only if the enumeration runs to completion within its budget and finds
exactly one subset of $F_i\cap G_j$ with sum $D_{ij}$. Since $I^*_{ij}$ is such a subset, the fixed subset is
$I^*_{ij}$.
\end{proof}

\begin{proposition}[two-line bound]\label{prop:twoline}
Let $a$ and $b$ be the multisets of the nonzero entries of two lines, with $p$ and $q$ elements. A set of these elements
is \emph{balanced} if its elements from $a$ and its elements from $b$ have equal sums. A \emph{block partition} of
$(a,b)$ splits the $p+q$ elements into nonempty balanced sets, the \emph{blocks}. Let the block number $\gamma^*(a,b)$
be the largest number of blocks of a block partition. Then every decomposition of $D$, with integer or real weights, has
at least $p+q-\gamma^*(a,b)$ terms, that is, $k^*(D)\ge k^*_{\R}(D)\ge p+q-\gamma^*(a,b)$.
\end{proposition}
\begin{proof}
Take a decomposition with $k$ terms and real weights $c_l>0$. Let $H$ be the matrix whose rows are indexed by the $p$
positions of the first line with nonzero entries and whose columns are indexed by the $q$ positions of the second line
with nonzero entries, with $H_{ef}=\sum_{l:\,\pi_l\text{ covers }e\text{ and }f}c_l$. Every term contributes to exactly
one entry of $H$, so $H$ has at most $k$ nonzero entries. The nonzero entries of $H$ form a bipartite graph on the $p+q$
positions without isolated vertices, because all these positions have nonzero entries. If this graph has $\gamma$
connected components, it has at least $p+q-\gamma$ edges. The row and column sums of $H$ are the entries of $a$ and $b$,
so the components form a block partition, and $\gamma\le\gamma^*(a,b)$. If $a$ and $b$ are a row and a column, their
common position appears on both sides, and the argument is the same.
\end{proof}

\begin{lemma}[equal values first]\label{lem:equal}
If a value $x$ occurs in both $a$ and $b$, with copies $x_a\in a$ and $x_b\in b$, then some block partition with
$\gamma^*(a,b)$ blocks contains the block $\{x_a,x_b\}$. Consequently, $\gamma^*(a,b)=M_2+\gamma^*(a',b')$, where $M_2$
is the number of pairs in a maximum matching of equal values, and $a'$ and $b'$ are the unmatched values of $a$ and $b$.
\end{lemma}
\begin{proof}
Take a block partition with $\gamma^*(a,b)$ blocks. If $x_a$ and $x_b$ lie in the same block $B$, then
$B\setminus\{x_a,x_b\}$ is balanced. Since all values are positive, this set is either empty or has elements on both
sides, so splitting it off does not decrease the number of blocks. If $x_a\in B_1$ and $x_b\in B_2$ with $B_1\ne B_2$,
replace $B_1$ and $B_2$ by $\{x_a,x_b\}$ and $(B_1\cup B_2)\setminus\{x_a,x_b\}$. The second set is balanced, and it is
nonempty, because otherwise $B_1=\{x_a\}$, which is not balanced since $x>0$. The other blocks do not change, so the
argument can be repeated for every pair.
\end{proof}

\begin{proposition}[exact block number]\label{prop:blockcount}
Let $m\ge1$, and let $v_1,\dots,v_m$ be the signed unmatched values, with sign $+$ for the values of $a'$ and sign $-$
for the values of $b'$. Their sum is zero, because both lines sum to $s$ and every matched pair removes the same value
from both.
\begin{enumerate}
\item The block number $\gamma^*(a',b')$ is the largest number of nonempty zero-sum prefixes over all orderings of
  $v_1,\dots,v_m$. Consequently, the recursion $\mathrm{dp}(S)=\max_{i\in S}\mathrm{dp}(S\setminus\{i\})+[\sum_{j\in
  S}v_j=0]$ with $\mathrm{dp}(\emptyset)=0$ gives $\gamma^*(a',b')=\mathrm{dp}(\{1,\dots,m\})$, where the bracket
  $[\,\cdot\,]$ is $1$ if the condition in it holds and $0$ otherwise.
\item Let $\mathcal Z$ be the family of nonempty zero-sum subsets of $\{1,\dots,m\}$. For a nonempty zero-sum set $S$,
  let $\min S$ be its element with the smallest index, and let $g(S)=\max\bigl(1,\max\{1+g(S\setminus Z):Z\in\mathcal
  Z,\ \min S\in Z\subsetneq S\}\bigr)$. Then $\gamma^*(a',b')=g(\{1,\dots,m\})$.
\end{enumerate}
\end{proposition}
\begin{proof}
\begin{enumerate}
\item A block partition, listed block by block, gives an ordering with a zero-sum prefix after every block. Conversely,
  consecutive zero-sum prefixes cut an ordering into balanced blocks.
\item In a block partition of $S$, some block contains $\min S$. This block is $S$ itself or a proper zero-sum subset
  $Z$, and $S\setminus Z$ is again a nonempty zero-sum set. This gives $\gamma^*(a',b')\le g(\{1,\dots,m\})$.
  Conversely, every value of the recursion is the number of blocks of a block partition built from $Z$ and a block
  partition of $S\setminus Z$, so equality holds. Any fixed choice of the distinguished element would do.
\end{enumerate}
\end{proof}

\begin{proposition}[packing bound]\label{prop:packing}
Let $a'$ and $b'$ be the unmatched values of Lemma~\ref{lem:equal}, so that no value occurs in both, and let
$m=|a'|+|b'|$. Then every block of unmatched values has at least three elements. Let $\mathcal Z_T$ be the family of
zero-sum subsets of the signed unmatched values with at least three and at most $T\ge3$ elements, and let
$w(Z)=1-|Z|/(T+1)$. For every $u\ge0$ with $\sum_{i\in Z}u_i\ge w(Z)$ for all $Z\in\mathcal Z_T$,
\[ \gamma^*(a',b')\le\Bigl\lfloor\frac{m}{T+1}+\sum_iu_i\Bigr\rfloor. \]
\end{proposition}
\begin{proof}
A single element has a nonzero sum, two elements of the same line have the same sign, and two equal values of different
lines have been matched before (Lemma~\ref{lem:equal}). Hence every block has at least three elements. Consider a block
partition with $\gamma$ blocks, and let $\mathcal P$ be the set of its blocks with at most $T$ elements. These blocks
belong to $\mathcal Z_T$ and are disjoint, and each of the other blocks has at least $T+1$ elements. Hence
\begin{align*}
\gamma&\le|\mathcal P|+\frac{m-|\bigcup\mathcal P|}{T+1}=\frac{m}{T+1}+\sum_{Z\in\mathcal P}w(Z)\\
&\le\frac{m}{T+1}+\sum_{Z\in\mathcal P}\sum_{i\in Z}u_i\le\frac{m}{T+1}+\sum_iu_i,
\end{align*}
where the last inequality holds because the blocks in $\mathcal P$ are disjoint and $u\ge0$. Since $\gamma^*(a',b')$ is
an integer, the bound may be rounded down.
\end{proof}

\begin{proposition}[SAT model]\label{prop:sat}
For $k\le s$, the SAT model for $k$ terms is satisfiable if and only if $D$ has a decomposition with precision $1/s$ and
at most $k$ distinct permutations. Consequently, if the model for $k-1\le s$ terms is unsatisfiable, then $k^*(D)\ge k$.
\end{proposition}
\begin{proof}
A satisfying assignment gives $k$ permutations and $k$ integer weights $c_l\ge1$ such that, for every position $e$, the
weights of the terms covering $e$ sum to $D_e$. This is a decomposition with exactly $k$ terms, possibly with equal
permutations, and merging the terms with equal permutations leaves at most $k$ distinct permutations. Conversely, given
a decomposition with $k'<k$ terms, some weight is at least two because $\sum_lc_l=s\ge k>k'$, and we split its term into
two terms with the same permutation. Repeating this until there are $k$ terms and sorting the weights gives a satisfying
assignment.
\end{proof}

\begin{proposition}[case split]\label{prop:cubes}
For a position $e=(i,j)$ with $D_e>0$, let $U_e$ be the largest value of $\min_{i'}D_{i'\pi(i')}$ over the permutations
$\pi$ with $\pi(i)=j$. For a line with nonzero positions $e_1,\dots,e_q$, we call integers $r_1,\dots,r_q$ with $\lceil
D_{e_t}/U_{e_t}\rceil\le r_t\le D_{e_t}$ and $\sum_tr_t=k$ an \emph{admissible composition} of $k$ into $q$ parts. The
case of a composition $(r_1,\dots,r_q)$ fixes the terms $1,\dots,r_1$ to position $e_1$, the next $r_2$ terms to
position $e_2$, and so on, and sorts the weights within each group instead of sorting all weights. Then $D$ has a
decomposition with precision $1/s$ and $k$ terms if and only if the case of some admissible composition is satisfiable.
\end{proposition}
\begin{proof}
Given a decomposition with $k$ terms, let $r_t$ be the number of terms that cover $e_t$. Each term covers exactly one
position of the line, so $\sum_tr_t=k$. Every weight is at least one, so $r_t\le D_{e_t}$. A term $(c_l,\pi_l)$ that
covers $e_t$ satisfies $c_l\le\min_{i'}D_{i'\pi_l(i')}\le U_{e_t}$, because the other terms contribute nonnegative
amounts to every entry that $\pi_l$ covers, so $D_{e_t}\le r_tU_{e_t}$ and $r_t\ge\lceil D_{e_t}/U_{e_t}\rceil$. This
bound also holds for real weights. Ordering the terms by their position in the line, and by weight within each position,
gives a satisfying assignment of the case of $(r_1,\dots,r_q)$. Conversely, every satisfying assignment of a case is a
decomposition with $k$ terms.
\end{proof}

\begin{proposition}[zeroing order]\label{prop:zeroing}
Let $(c_l,\pi_l)_{l=1}^k$ be a decomposition of $D$ with a zeroing order $l_1,\dots,l_k$, and let
$R_t=\sum_{u=t}^kc_{l_u}P_{\pi_{l_u}}$ for $t=1,\dots,k$. Then $c_{l_t}=\min_i(R_t)_{i\pi_{l_t}(i)}$ for every $t$, that
is, the decomposition is the sequence of the zeroing steps with the permutations $\pi_{l_1},\dots,\pi_{l_k}$, applied to
$D=R_1$.
\end{proposition}
\begin{proof}
Since $R_t\ge c_{l_t}P_{\pi_{l_t}}$ entrywise, the smallest entry of $R_t$ at the positions covered by $\pi_{l_t}$ is at
least $c_{l_t}$. At a position that is covered by $\pi_{l_t}$ and by no later term, the entry of $R_t$ is exactly
$c_{l_t}$.
\end{proof}

\section{Test instances}\label{app:instances}

Each test instance is generated from $k_0$ random permutations, with $k_0=n$ for sparse and $k_0=n^2$ for dense
instances~\cite{koch2025}. The weights are multiples of $1/s$, namely the gaps between $k_0-1$ sorted random multiples
of $1/s$ in $[0,1]$, where $s=10^d$ and $d$ is the number of decimal digits of $n^2$ plus two. Table~\ref{tab:stats}
gives statistics of the \emph{generating decompositions}, made of the generating permutations and their weights, which
the instance files contain. These statistics only describe the instances, and RCS reads none of these data
(Section~\ref{sec:implementation}). In sparse instances, a nonzero position is covered by about $1.6$ generating
permutations on average, close to the limit $1/(1-\mathrm{e}^{-1})\approx1.58$, and more than $90\,\%$ of the generating
permutations have a position that no other generating permutation covers, a \emph{witness position}. In dense instances,
a position is covered by about $n$ generating permutations, and the share of generating permutations with a witness
position falls from $54\,\%$ for $n=3$ to $1\,\%$ for $n=6$. For $n\ge7$, no generating permutation has one.

\begin{table}[ht]\centering\small
\caption{Statistics of the test instances, as means per size and class, computed from the permutations and weights from
which each matrix was generated. The columns give the scale $s$ of the scaled matrix $D=sA$, the number of nonzero
entries per row or column, the number of generating permutations through a nonzero position, and the share of generating
permutations with a position that no other one covers.}\label{tab:stats}
\StatsTable
\end{table}

\FloatBarrier
\section{Results per family}\label{app:results}

A \emph{family} is the set of instances of the same size and class, and Table~\ref{tab:family} gives the results per
family. The lower bounds of the dense families with $4\le n\le9$ include SAT certificates, which hold only for weights
with precision $1/s$, and all other bounds also hold for real weights. Appendix~\ref{app:all} gives the results for every
instance.

\begin{table}[ht]\centering\small
\caption{Results per family of ten instances (five for the dense family with $n=100$). The columns are as in
Table~\ref{tab:groups}, but the length and the bound are not divided by $n$.}\label{tab:family}
\FamilyTable
\end{table}

\FloatBarrier
\section{Results per instance}\label{app:all}

Table~\ref{tab:all} gives, for every test instance, the lower bound on the minimum length, the length of the best decomposition found in five independent runs, the gap between the two, whether optimality is proven, the number of runs that reached this length, the time from the start of a run until it first found a decomposition of its final length and the running time of the run (both rounds), in seconds and as means over the five runs. A run stops early when its decomposition is proven optimal. A bound from a SAT certificate (dense matrices with $n\le9$) holds only for weights with precision $1/s$, and all other bounds also hold for real weights.

\begingroup
\renewcommand{\tablename}{Table}
{\small\begin{longtable}{lrrrcrrr}
\caption{Lower bound, length of the best decomposition, gap, proof of optimality, time to the best decomposition and running time for every test instance.}\label{tab:all}\\
\toprule
instance & bound & length & gap & proven & runs & time to best [s] & running time [s] \\ \midrule\endhead
\texttt{B3\_3\_1} & 2 & 2 & 0 & $\checkmark$ & 5/5 & 0.00 & 0.00 \\
\texttt{B3\_3\_2} & 3 & 3 & 0 & $\checkmark$ & 5/5 & 0.00 & 0.00 \\
\texttt{B3\_3\_3} & 3 & 3 & 0 & $\checkmark$ & 5/5 & 0.00 & 0.00 \\
\texttt{B3\_3\_4} & 3 & 3 & 0 & $\checkmark$ & 5/5 & 0.00 & 0.00 \\
\texttt{B3\_3\_5} & 3 & 3 & 0 & $\checkmark$ & 5/5 & 0.00 & 0.00 \\
\texttt{B3\_3\_6} & 3 & 3 & 0 & $\checkmark$ & 5/5 & 0.00 & 0.00 \\
\texttt{B3\_3\_7} & 2 & 2 & 0 & $\checkmark$ & 5/5 & 0.00 & 0.00 \\
\texttt{B3\_3\_8} & 3 & 3 & 0 & $\checkmark$ & 5/5 & 0.00 & 0.00 \\
\texttt{B3\_3\_9} & 2 & 2 & 0 & $\checkmark$ & 5/5 & 0.00 & 0.00 \\
\texttt{B3\_3\_10} & 3 & 3 & 0 & $\checkmark$ & 5/5 & 0.00 & 0.00 \\
\midrule
\texttt{B4\_4\_1} & 4 & 4 & 0 & $\checkmark$ & 5/5 & 0.00 & 0.00 \\
\texttt{B4\_4\_2} & 4 & 4 & 0 & $\checkmark$ & 5/5 & 0.00 & 0.00 \\
\texttt{B4\_4\_3} & 3 & 3 & 0 & $\checkmark$ & 5/5 & 0.00 & 0.00 \\
\texttt{B4\_4\_4} & 3 & 3 & 0 & $\checkmark$ & 5/5 & 0.00 & 0.00 \\
\texttt{B4\_4\_5} & 4 & 4 & 0 & $\checkmark$ & 5/5 & 0.00 & 0.00 \\
\texttt{B4\_4\_6} & 4 & 4 & 0 & $\checkmark$ & 5/5 & 0.00 & 0.00 \\
\texttt{B4\_4\_7} & 4 & 4 & 0 & $\checkmark$ & 5/5 & 0.00 & 0.00 \\
\texttt{B4\_4\_8} & 4 & 4 & 0 & $\checkmark$ & 5/5 & 0.00 & 0.00 \\
\texttt{B4\_4\_9} & 4 & 4 & 0 & $\checkmark$ & 5/5 & 0.00 & 0.00 \\
\texttt{B4\_4\_10} & 4 & 4 & 0 & $\checkmark$ & 5/5 & 0.00 & 0.00 \\
\midrule
\texttt{B5\_5\_1} & 5 & 5 & 0 & $\checkmark$ & 5/5 & 0.00 & 0.00 \\
\texttt{B5\_5\_2} & 5 & 5 & 0 & $\checkmark$ & 5/5 & 0.00 & 0.00 \\
\texttt{B5\_5\_3} & 5 & 5 & 0 & $\checkmark$ & 5/5 & 0.00 & 0.00 \\
\texttt{B5\_5\_4} & 5 & 5 & 0 & $\checkmark$ & 5/5 & 0.00 & 0.00 \\
\texttt{B5\_5\_5} & 5 & 5 & 0 & $\checkmark$ & 5/5 & 0.00 & 0.00 \\
\texttt{B5\_5\_6} & 5 & 5 & 0 & $\checkmark$ & 5/5 & 0.00 & 0.00 \\
\texttt{B5\_5\_7} & 5 & 5 & 0 & $\checkmark$ & 5/5 & 0.00 & 0.00 \\
\texttt{B5\_5\_8} & 5 & 5 & 0 & $\checkmark$ & 5/5 & 0.00 & 0.00 \\
\texttt{B5\_5\_9} & 3 & 3 & 0 & $\checkmark$ & 5/5 & 0.00 & 0.00 \\
\texttt{B5\_5\_10} & 5 & 5 & 0 & $\checkmark$ & 5/5 & 0.00 & 0.00 \\
\midrule
\texttt{B6\_6\_1} & 6 & 6 & 0 & $\checkmark$ & 5/5 & 0.00 & 0.00 \\
\texttt{B6\_6\_2} & 6 & 6 & 0 & $\checkmark$ & 5/5 & 0.00 & 0.00 \\
\texttt{B6\_6\_3} & 6 & 6 & 0 & $\checkmark$ & 5/5 & 0.00 & 0.00 \\
\texttt{B6\_6\_4} & 6 & 6 & 0 & $\checkmark$ & 5/5 & 0.00 & 0.00 \\
\texttt{B6\_6\_5} & 6 & 6 & 0 & $\checkmark$ & 5/5 & 0.00 & 0.00 \\
\texttt{B6\_6\_6} & 6 & 6 & 0 & $\checkmark$ & 5/5 & 0.00 & 0.00 \\
\texttt{B6\_6\_7} & 6 & 6 & 0 & $\checkmark$ & 5/5 & 0.00 & 0.00 \\
\texttt{B6\_6\_8} & 6 & 6 & 0 & $\checkmark$ & 5/5 & 0.00 & 0.00 \\
\texttt{B6\_6\_9} & 6 & 6 & 0 & $\checkmark$ & 5/5 & 0.00 & 0.00 \\
\texttt{B6\_6\_10} & 6 & 6 & 0 & $\checkmark$ & 5/5 & 0.00 & 0.00 \\
\midrule
\texttt{B7\_7\_1} & 7 & 7 & 0 & $\checkmark$ & 5/5 & 0.00 & 0.00 \\
\texttt{B7\_7\_2} & 7 & 7 & 0 & $\checkmark$ & 5/5 & 0.00 & 0.00 \\
\texttt{B7\_7\_3} & 7 & 7 & 0 & $\checkmark$ & 5/5 & 0.00 & 0.00 \\
\texttt{B7\_7\_4} & 7 & 7 & 0 & $\checkmark$ & 5/5 & 0.00 & 0.00 \\
\texttt{B7\_7\_5} & 7 & 7 & 0 & $\checkmark$ & 5/5 & 0.00 & 0.00 \\
\texttt{B7\_7\_6} & 7 & 7 & 0 & $\checkmark$ & 5/5 & 0.00 & 0.01 \\
\texttt{B7\_7\_7} & 7 & 7 & 0 & $\checkmark$ & 5/5 & 0.00 & 0.01 \\
\texttt{B7\_7\_8} & 7 & 7 & 0 & $\checkmark$ & 5/5 & 0.00 & 0.00 \\
\texttt{B7\_7\_9} & 7 & 7 & 0 & $\checkmark$ & 5/5 & 0.00 & 0.00 \\
\texttt{B7\_7\_10} & 7 & 7 & 0 & $\checkmark$ & 5/5 & 0.00 & 0.01 \\
\midrule
\texttt{B8\_8\_1} & 8 & 8 & 0 & $\checkmark$ & 5/5 & 0.00 & 0.00 \\
\texttt{B8\_8\_2} & 8 & 8 & 0 & $\checkmark$ & 5/5 & 0.00 & 0.00 \\
\texttt{B8\_8\_3} & 8 & 8 & 0 & $\checkmark$ & 5/5 & 0.00 & 0.00 \\
\texttt{B8\_8\_4} & 8 & 8 & 0 & $\checkmark$ & 5/5 & 0.00 & 0.00 \\
\texttt{B8\_8\_5} & 8 & 8 & 0 & $\checkmark$ & 5/5 & 0.00 & 0.00 \\
\texttt{B8\_8\_6} & 8 & 8 & 0 & $\checkmark$ & 5/5 & 0.00 & 0.00 \\
\texttt{B8\_8\_7} & 8 & 8 & 0 & $\checkmark$ & 5/5 & 0.00 & 0.00 \\
\texttt{B8\_8\_8} & 8 & 8 & 0 & $\checkmark$ & 5/5 & 0.00 & 0.00 \\
\texttt{B8\_8\_9} & 8 & 8 & 0 & $\checkmark$ & 5/5 & 0.00 & 0.00 \\
\texttt{B8\_8\_10} & 8 & 8 & 0 & $\checkmark$ & 5/5 & 0.00 & 0.00 \\
\midrule
\texttt{B9\_9\_1} & 9 & 9 & 0 & $\checkmark$ & 5/5 & 0.00 & 0.00 \\
\texttt{B9\_9\_2} & 9 & 9 & 0 & $\checkmark$ & 5/5 & 0.00 & 0.00 \\
\texttt{B9\_9\_3} & 9 & 9 & 0 & $\checkmark$ & 5/5 & 0.00 & 0.00 \\
\texttt{B9\_9\_4} & 9 & 9 & 0 & $\checkmark$ & 5/5 & 0.00 & 0.00 \\
\texttt{B9\_9\_5} & 9 & 9 & 0 & $\checkmark$ & 5/5 & 0.00 & 0.00 \\
\texttt{B9\_9\_6} & 9 & 9 & 0 & $\checkmark$ & 5/5 & 0.00 & 0.00 \\
\texttt{B9\_9\_7} & 9 & 9 & 0 & $\checkmark$ & 5/5 & 0.01 & 0.01 \\
\texttt{B9\_9\_8} & 9 & 9 & 0 & $\checkmark$ & 5/5 & 0.00 & 0.00 \\
\texttt{B9\_9\_9} & 9 & 9 & 0 & $\checkmark$ & 5/5 & 0.00 & 0.00 \\
\texttt{B9\_9\_10} & 9 & 9 & 0 & $\checkmark$ & 5/5 & 0.00 & 0.00 \\
\midrule
\texttt{B10\_10\_1} & 10 & 10 & 0 & $\checkmark$ & 5/5 & 0.00 & 0.00 \\
\texttt{B10\_10\_2} & 10 & 10 & 0 & $\checkmark$ & 5/5 & 0.00 & 0.00 \\
\texttt{B10\_10\_3} & 10 & 10 & 0 & $\checkmark$ & 5/5 & 0.00 & 0.00 \\
\texttt{B10\_10\_4} & 10 & 10 & 0 & $\checkmark$ & 5/5 & 0.00 & 0.00 \\
\texttt{B10\_10\_5} & 10 & 10 & 0 & $\checkmark$ & 5/5 & 0.00 & 0.00 \\
\texttt{B10\_10\_6} & 10 & 10 & 0 & $\checkmark$ & 5/5 & 0.00 & 0.00 \\
\texttt{B10\_10\_7} & 10 & 10 & 0 & $\checkmark$ & 5/5 & 0.00 & 0.00 \\
\texttt{B10\_10\_8} & 10 & 10 & 0 & $\checkmark$ & 5/5 & 0.00 & 0.00 \\
\texttt{B10\_10\_9} & 10 & 10 & 0 & $\checkmark$ & 5/5 & 0.00 & 0.00 \\
\texttt{B10\_10\_10} & 10 & 10 & 0 & $\checkmark$ & 5/5 & 0.00 & 0.00 \\
\midrule
\texttt{B11\_11\_1} & 11 & 11 & 0 & $\checkmark$ & 5/5 & 0.00 & 0.00 \\
\texttt{B11\_11\_2} & 11 & 11 & 0 & $\checkmark$ & 5/5 & 0.00 & 0.00 \\
\texttt{B11\_11\_3} & 11 & 11 & 0 & $\checkmark$ & 5/5 & 0.00 & 0.00 \\
\texttt{B11\_11\_4} & 11 & 11 & 0 & $\checkmark$ & 5/5 & 0.00 & 0.00 \\
\texttt{B11\_11\_5} & 11 & 11 & 0 & $\checkmark$ & 5/5 & 0.00 & 0.00 \\
\texttt{B11\_11\_6} & 11 & 11 & 0 & $\checkmark$ & 5/5 & 0.00 & 0.00 \\
\texttt{B11\_11\_7} & 11 & 11 & 0 & $\checkmark$ & 5/5 & 0.00 & 0.00 \\
\texttt{B11\_11\_8} & 11 & 11 & 0 & $\checkmark$ & 5/5 & 0.00 & 0.00 \\
\texttt{B11\_11\_9} & 11 & 11 & 0 & $\checkmark$ & 5/5 & 0.00 & 0.00 \\
\texttt{B11\_11\_10} & 11 & 11 & 0 & $\checkmark$ & 5/5 & 0.00 & 0.00 \\
\midrule
\texttt{B12\_12\_1} & 12 & 12 & 0 & $\checkmark$ & 5/5 & 0.00 & 0.00 \\
\texttt{B12\_12\_2} & 12 & 12 & 0 & $\checkmark$ & 5/5 & 0.00 & 0.00 \\
\texttt{B12\_12\_3} & 12 & 12 & 0 & $\checkmark$ & 5/5 & 0.00 & 0.00 \\
\texttt{B12\_12\_4} & 12 & 12 & 0 & $\checkmark$ & 5/5 & 0.00 & 0.00 \\
\texttt{B12\_12\_5} & 12 & 12 & 0 & $\checkmark$ & 5/5 & 0.00 & 0.00 \\
\texttt{B12\_12\_6} & 12 & 12 & 0 & $\checkmark$ & 5/5 & 0.00 & 0.00 \\
\texttt{B12\_12\_7} & 12 & 12 & 0 & $\checkmark$ & 5/5 & 0.00 & 0.00 \\
\texttt{B12\_12\_8} & 12 & 12 & 0 & $\checkmark$ & 5/5 & 0.00 & 0.00 \\
\texttt{B12\_12\_9} & 12 & 12 & 0 & $\checkmark$ & 5/5 & 0.00 & 0.01 \\
\texttt{B12\_12\_10} & 12 & 12 & 0 & $\checkmark$ & 5/5 & 0.00 & 0.00 \\
\midrule
\texttt{B13\_13\_1} & 13 & 13 & 0 & $\checkmark$ & 5/5 & 0.00 & 0.00 \\
\texttt{B13\_13\_2} & 13 & 13 & 0 & $\checkmark$ & 5/5 & 0.00 & 0.00 \\
\texttt{B13\_13\_3} & 13 & 13 & 0 & $\checkmark$ & 5/5 & 0.00 & 0.00 \\
\texttt{B13\_13\_4} & 13 & 13 & 0 & $\checkmark$ & 5/5 & 0.00 & 0.00 \\
\texttt{B13\_13\_5} & 13 & 13 & 0 & $\checkmark$ & 5/5 & 0.00 & 0.00 \\
\texttt{B13\_13\_6} & 13 & 13 & 0 & $\checkmark$ & 5/5 & 0.00 & 0.00 \\
\texttt{B13\_13\_7} & 13 & 13 & 0 & $\checkmark$ & 5/5 & 0.00 & 0.00 \\
\texttt{B13\_13\_8} & 13 & 13 & 0 & $\checkmark$ & 5/5 & 0.00 & 0.00 \\
\texttt{B13\_13\_9} & 13 & 13 & 0 & $\checkmark$ & 5/5 & 0.00 & 0.00 \\
\texttt{B13\_13\_10} & 13 & 13 & 0 & $\checkmark$ & 5/5 & 0.00 & 0.00 \\
\midrule
\texttt{B14\_14\_1} & 14 & 14 & 0 & $\checkmark$ & 5/5 & 0.00 & 0.00 \\
\texttt{B14\_14\_2} & 14 & 14 & 0 & $\checkmark$ & 5/5 & 0.00 & 0.00 \\
\texttt{B14\_14\_3} & 14 & 14 & 0 & $\checkmark$ & 5/5 & 0.00 & 0.00 \\
\texttt{B14\_14\_4} & 14 & 14 & 0 & $\checkmark$ & 5/5 & 0.00 & 0.00 \\
\texttt{B14\_14\_5} & 14 & 14 & 0 & $\checkmark$ & 5/5 & 0.00 & 0.00 \\
\texttt{B14\_14\_6} & 14 & 14 & 0 & $\checkmark$ & 5/5 & 0.00 & 0.00 \\
\texttt{B14\_14\_7} & 14 & 14 & 0 & $\checkmark$ & 5/5 & 0.00 & 0.00 \\
\texttt{B14\_14\_8} & 14 & 14 & 0 & $\checkmark$ & 5/5 & 0.00 & 0.00 \\
\texttt{B14\_14\_9} & 14 & 14 & 0 & $\checkmark$ & 5/5 & 0.00 & 0.00 \\
\texttt{B14\_14\_10} & 14 & 14 & 0 & $\checkmark$ & 5/5 & 0.00 & 0.00 \\
\midrule
\texttt{B15\_15\_1} & 15 & 15 & 0 & $\checkmark$ & 5/5 & 0.00 & 0.00 \\
\texttt{B15\_15\_2} & 15 & 15 & 0 & $\checkmark$ & 5/5 & 0.00 & 0.01 \\
\texttt{B15\_15\_3} & 15 & 15 & 0 & $\checkmark$ & 5/5 & 0.00 & 0.00 \\
\texttt{B15\_15\_4} & 15 & 15 & 0 & $\checkmark$ & 5/5 & 0.00 & 0.00 \\
\texttt{B15\_15\_5} & 15 & 15 & 0 & $\checkmark$ & 5/5 & 0.00 & 0.00 \\
\texttt{B15\_15\_6} & 15 & 15 & 0 & $\checkmark$ & 5/5 & 0.00 & 0.00 \\
\texttt{B15\_15\_7} & 15 & 15 & 0 & $\checkmark$ & 5/5 & 0.00 & 0.00 \\
\texttt{B15\_15\_8} & 15 & 15 & 0 & $\checkmark$ & 5/5 & 0.00 & 0.00 \\
\texttt{B15\_15\_9} & 15 & 15 & 0 & $\checkmark$ & 5/5 & 0.00 & 0.02 \\
\texttt{B15\_15\_10} & 15 & 15 & 0 & $\checkmark$ & 5/5 & 0.00 & 0.00 \\
\midrule
\texttt{B16\_16\_1} & 16 & 16 & 0 & $\checkmark$ & 5/5 & 0.00 & 0.00 \\
\texttt{B16\_16\_2} & 16 & 16 & 0 & $\checkmark$ & 5/5 & 0.00 & 0.01 \\
\texttt{B16\_16\_3} & 16 & 16 & 0 & $\checkmark$ & 5/5 & 0.00 & 0.00 \\
\texttt{B16\_16\_4} & 16 & 16 & 0 & $\checkmark$ & 5/5 & 0.00 & 0.02 \\
\texttt{B16\_16\_5} & 16 & 16 & 0 & $\checkmark$ & 5/5 & 0.00 & 0.00 \\
\texttt{B16\_16\_6} & 16 & 16 & 0 & $\checkmark$ & 5/5 & 0.00 & 0.00 \\
\texttt{B16\_16\_7} & 16 & 16 & 0 & $\checkmark$ & 5/5 & 0.00 & 0.00 \\
\texttt{B16\_16\_8} & 16 & 16 & 0 & $\checkmark$ & 5/5 & 0.00 & 0.02 \\
\texttt{B16\_16\_9} & 16 & 16 & 0 & $\checkmark$ & 5/5 & 0.00 & 0.02 \\
\texttt{B16\_16\_10} & 16 & 16 & 0 & $\checkmark$ & 5/5 & 0.00 & 0.03 \\
\midrule
\texttt{B24\_24\_1} & 24 & 24 & 0 & $\checkmark$ & 5/5 & 0.00 & 0.01 \\
\texttt{B24\_24\_2} & 24 & 24 & 0 & $\checkmark$ & 5/5 & 0.00 & 0.00 \\
\texttt{B24\_24\_3} & 24 & 24 & 0 & $\checkmark$ & 5/5 & 0.00 & 0.00 \\
\texttt{B24\_24\_4} & 24 & 24 & 0 & $\checkmark$ & 5/5 & 0.00 & 0.00 \\
\texttt{B24\_24\_5} & 24 & 24 & 0 & $\checkmark$ & 5/5 & 0.00 & 0.00 \\
\texttt{B24\_24\_6} & 24 & 24 & 0 & $\checkmark$ & 5/5 & 0.00 & 0.00 \\
\texttt{B24\_24\_7} & 24 & 24 & 0 & $\checkmark$ & 5/5 & 0.00 & 0.00 \\
\texttt{B24\_24\_8} & 24 & 24 & 0 & $\checkmark$ & 5/5 & 0.00 & 0.00 \\
\texttt{B24\_24\_9} & 24 & 24 & 0 & $\checkmark$ & 5/5 & 0.00 & 0.00 \\
\texttt{B24\_24\_10} & 24 & 24 & 0 & $\checkmark$ & 5/5 & 0.00 & 0.00 \\
\midrule
\texttt{B32\_32\_1} & 32 & 32 & 0 & $\checkmark$ & 5/5 & 0.00 & 0.01 \\
\texttt{B32\_32\_2} & 32 & 32 & 0 & $\checkmark$ & 5/5 & 0.00 & 1.23 \\
\texttt{B32\_32\_3} & 32 & 32 & 0 & $\checkmark$ & 5/5 & 0.00 & 0.00 \\
\texttt{B32\_32\_4} & 32 & 32 & 0 & $\checkmark$ & 5/5 & 0.00 & 0.00 \\
\texttt{B32\_32\_5} & 32 & 32 & 0 & $\checkmark$ & 5/5 & 0.00 & 0.00 \\
\texttt{B32\_32\_6} & 32 & 32 & 0 & $\checkmark$ & 5/5 & 0.00 & 0.00 \\
\texttt{B32\_32\_7} & 32 & 32 & 0 & $\checkmark$ & 5/5 & 0.00 & 0.00 \\
\texttt{B32\_32\_8} & 32 & 32 & 0 & $\checkmark$ & 5/5 & 0.00 & 0.00 \\
\texttt{B32\_32\_9} & 32 & 32 & 0 & $\checkmark$ & 5/5 & 0.00 & 0.00 \\
\texttt{B32\_32\_10} & 32 & 32 & 0 & $\checkmark$ & 5/5 & 0.00 & 0.00 \\
\midrule
\texttt{B48\_48\_1} & 48 & 48 & 0 & $\checkmark$ & 5/5 & 0.00 & 0.02 \\
\texttt{B48\_48\_2} & 48 & 48 & 0 & $\checkmark$ & 5/5 & 0.00 & 0.01 \\
\texttt{B48\_48\_3} & 48 & 48 & 0 & $\checkmark$ & 5/5 & 0.01 & 0.03 \\
\texttt{B48\_48\_4} & 48 & 48 & 0 & $\checkmark$ & 5/5 & 0.00 & 0.01 \\
\texttt{B48\_48\_5} & 48 & 48 & 0 & $\checkmark$ & 5/5 & 0.02 & 0.03 \\
\texttt{B48\_48\_6} & 48 & 48 & 0 & $\checkmark$ & 5/5 & 0.00 & 0.02 \\
\texttt{B48\_48\_7} & 48 & 48 & 0 & $\checkmark$ & 5/5 & 0.02 & 0.02 \\
\texttt{B48\_48\_8} & 48 & 48 & 0 & $\checkmark$ & 5/5 & 0.00 & 0.02 \\
\texttt{B48\_48\_9} & 48 & 48 & 0 & $\checkmark$ & 5/5 & 0.00 & 0.01 \\
\texttt{B48\_48\_10} & 48 & 48 & 0 & $\checkmark$ & 5/5 & 0.00 & 0.01 \\
\midrule
\texttt{B64\_64\_1} & 64 & 64 & 0 & $\checkmark$ & 5/5 & 0.05 & 0.07 \\
\texttt{B64\_64\_2} & 64 & 64 & 0 & $\checkmark$ & 5/5 & 0.01 & 0.04 \\
\texttt{B64\_64\_3} & 64 & 64 & 0 & $\checkmark$ & 5/5 & 0.01 & 0.05 \\
\texttt{B64\_64\_4} & 64 & 64 & 0 & $\checkmark$ & 5/5 & 0.01 & 0.11 \\
\texttt{B64\_64\_5} & 64 & 64 & 0 & $\checkmark$ & 5/5 & 0.01 & 0.33 \\
\texttt{B64\_64\_6} & 64 & 64 & 0 & $\checkmark$ & 5/5 & 0.01 & 0.13 \\
\texttt{B64\_64\_7} & 64 & 64 & 0 & $\checkmark$ & 5/5 & 0.01 & 0.12 \\
\texttt{B64\_64\_8} & 64 & 64 & 0 & $\checkmark$ & 5/5 & 0.01 & 0.20 \\
\texttt{B64\_64\_9} & 64 & 64 & 0 & $\checkmark$ & 5/5 & 0.01 & 0.23 \\
\texttt{B64\_64\_10} & 64 & 64 & 0 & $\checkmark$ & 5/5 & 0.01 & 4.09 \\
\midrule
\texttt{B100\_100\_1} & 100 & 100 & 0 & $\checkmark$ & 5/5 & 0.06 & 8.39 \\
\texttt{B100\_100\_2} & 100 & 100 & 0 & $\checkmark$ & 5/5 & 0.06 & 129.39 \\
\texttt{B100\_100\_3} & 100 & 100 & 0 & $\checkmark$ & 5/5 & 0.06 & 22.01 \\
\texttt{B100\_100\_4} & 100 & 100 & 0 & $\checkmark$ & 5/5 & 0.06 & 92.66 \\
\texttt{B100\_100\_5} & 100 & 100 & 0 & $\checkmark$ & 5/5 & 0.05 & 0.87 \\
\texttt{B100\_100\_6} & 100 & 100 & 0 & $\checkmark$ & 5/5 & 0.05 & 110.62 \\
\texttt{B100\_100\_7} & 99 & 100 & 1 &  & 5/5 & 0.06 & 660.00 \\
\texttt{B100\_100\_8} & 100 & 100 & 0 & $\checkmark$ & 5/5 & 0.05 & 126.83 \\
\texttt{B100\_100\_9} & 99 & 100 & 1 &  & 5/5 & 0.05 & 660.00 \\
\texttt{B100\_100\_10} & 99 & 100 & 1 &  & 5/5 & 0.05 & 660.00 \\
\midrule
\texttt{B3\_9\_1} & 5 & 5 & 0 & $\checkmark$ & 5/5 & 0.00 & 0.00 \\
\texttt{B3\_9\_2} & 5 & 5 & 0 & $\checkmark$ & 5/5 & 0.00 & 0.00 \\
\texttt{B3\_9\_3} & 5 & 5 & 0 & $\checkmark$ & 5/5 & 0.00 & 0.00 \\
\texttt{B3\_9\_4} & 5 & 5 & 0 & $\checkmark$ & 5/5 & 0.00 & 0.00 \\
\texttt{B3\_9\_5} & 4 & 4 & 0 & $\checkmark$ & 5/5 & 0.00 & 0.00 \\
\texttt{B3\_9\_6} & 5 & 5 & 0 & $\checkmark$ & 5/5 & 0.00 & 0.00 \\
\texttt{B3\_9\_7} & 4 & 4 & 0 & $\checkmark$ & 5/5 & 0.00 & 0.00 \\
\texttt{B3\_9\_8} & 4 & 4 & 0 & $\checkmark$ & 5/5 & 0.00 & 0.00 \\
\texttt{B3\_9\_9} & 4 & 4 & 0 & $\checkmark$ & 5/5 & 0.00 & 0.01 \\
\texttt{B3\_9\_10} & 4 & 4 & 0 & $\checkmark$ & 5/5 & 0.00 & 0.00 \\
\midrule
\texttt{B4\_16\_1} & 9 & 9 & 0 & $\checkmark$ & 5/5 & 0.00 & 0.07 \\
\texttt{B4\_16\_2} & 10 & 10 & 0 & $\checkmark$ & 5/5 & 0.00 & 0.29 \\
\texttt{B4\_16\_3} & 8 & 8 & 0 & $\checkmark$ & 5/5 & 0.00 & 0.04 \\
\texttt{B4\_16\_4} & 10 & 10 & 0 & $\checkmark$ & 5/5 & 0.00 & 0.15 \\
\texttt{B4\_16\_5} & 10 & 10 & 0 & $\checkmark$ & 5/5 & 0.00 & 0.22 \\
\texttt{B4\_16\_6} & 10 & 10 & 0 & $\checkmark$ & 5/5 & 0.00 & 0.22 \\
\texttt{B4\_16\_7} & 10 & 10 & 0 & $\checkmark$ & 5/5 & 0.00 & 0.32 \\
\texttt{B4\_16\_8} & 10 & 10 & 0 & $\checkmark$ & 5/5 & 0.00 & 0.51 \\
\texttt{B4\_16\_9} & 10 & 10 & 0 & $\checkmark$ & 5/5 & 0.00 & 0.25 \\
\texttt{B4\_16\_10} & 9 & 9 & 0 & $\checkmark$ & 5/5 & 0.00 & 0.09 \\
\midrule
\texttt{B5\_25\_1} & 14 & 14 & 0 & $\checkmark$ & 5/5 & 0.41 & 211.22 \\
\texttt{B5\_25\_2} & 13 & 14 & 1 &  & 5/5 & 0.03 & 660.73 \\
\texttt{B5\_25\_3} & 13 & 14 & 1 &  & 5/5 & 0.28 & 660.87 \\
\texttt{B5\_25\_4} & 12 & 13 & 1 &  & 5/5 & 512.08 & 660.79 \\
\texttt{B5\_25\_5} & 13 & 14 & 1 &  & 5/5 & 2.41 & 660.54 \\
\texttt{B5\_25\_6} & 12 & 14 & 2 &  & 5/5 & 1.81 & 660.48 \\
\texttt{B5\_25\_7} & 13 & 14 & 1 &  & 5/5 & 0.45 & 660.54 \\
\texttt{B5\_25\_8} & 13 & 14 & 1 &  & 5/5 & 221.33 & 660.50 \\
\texttt{B5\_25\_9} & 13 & 14 & 1 &  & 5/5 & 0.03 & 660.51 \\
\texttt{B5\_25\_10} & 13 & 14 & 1 &  & 5/5 & 1.22 & 660.48 \\
\midrule
\texttt{B6\_36\_1} & 14 & 19 & 5 &  & 5/5 & 40.39 & 660.30 \\
\texttt{B6\_36\_2} & 14 & 19 & 5 &  & 5/5 & 3.60 & 661.40 \\
\texttt{B6\_36\_3} & 13 & 19 & 6 &  & 5/5 & 55.91 & 660.51 \\
\texttt{B6\_36\_4} & 14 & 19 & 5 &  & 5/5 & 11.33 & 661.08 \\
\texttt{B6\_36\_5} & 13 & 19 & 6 &  & 5/5 & 64.63 & 660.60 \\
\texttt{B6\_36\_6} & 14 & 19 & 5 &  & 5/5 & 241.92 & 660.87 \\
\texttt{B6\_36\_7} & 13 & 19 & 6 &  & 5/5 & 6.96 & 660.30 \\
\texttt{B6\_36\_8} & 14 & 20 & 6 &  & 5/5 & 3.48 & 661.52 \\
\texttt{B6\_36\_9} & 15 & 20 & 5 &  & 5/5 & 3.93 & 660.62 \\
\texttt{B6\_36\_10} & 13 & 19 & 6 &  & 5/5 & 18.58 & 660.41 \\
\midrule
\texttt{B7\_49\_1} & 16 & 25 & 9 &  & 1/5 & 76.17 & 660.09 \\
\texttt{B7\_49\_2} & 15 & 25 & 10 &  & 2/5 & 209.47 & 660.10 \\
\texttt{B7\_49\_3} & 15 & 25 & 10 &  & 5/5 & 330.80 & 660.08 \\
\texttt{B7\_49\_4} & 15 & 25 & 10 &  & 2/5 & 127.01 & 660.09 \\
\texttt{B7\_49\_5} & 15 & 25 & 10 &  & 1/5 & 71.03 & 660.09 \\
\texttt{B7\_49\_6} & 15 & 25 & 10 &  & 5/5 & 237.82 & 660.09 \\
\texttt{B7\_49\_7} & 14 & 25 & 11 &  & 5/5 & 58.36 & 660.08 \\
\texttt{B7\_49\_8} & 15 & 24 & 9 &  & 2/5 & 115.83 & 660.08 \\
\texttt{B7\_49\_9} & 15 & 25 & 10 &  & 5/5 & 124.39 & 660.08 \\
\texttt{B7\_49\_10} & 15 & 25 & 10 &  & 3/5 & 271.44 & 660.07 \\
\midrule
\texttt{B8\_64\_1} & 16 & 30 & 14 &  & 5/5 & 48.05 & 660.26 \\
\texttt{B8\_64\_2} & 15 & 30 & 15 &  & 5/5 & 13.87 & 660.30 \\
\texttt{B8\_64\_3} & 16 & 29 & 13 &  & 1/5 & 266.22 & 660.32 \\
\texttt{B8\_64\_4} & 15 & 30 & 15 &  & 5/5 & 85.90 & 660.29 \\
\texttt{B8\_64\_5} & 16 & 30 & 14 &  & 5/5 & 19.50 & 660.29 \\
\texttt{B8\_64\_6} & 15 & 30 & 15 &  & 5/5 & 10.29 & 660.31 \\
\texttt{B8\_64\_7} & 16 & 30 & 14 &  & 5/5 & 16.88 & 660.31 \\
\texttt{B8\_64\_8} & 15 & 30 & 15 &  & 5/5 & 243.89 & 660.36 \\
\texttt{B8\_64\_9} & 15 & 30 & 15 &  & 5/5 & 112.38 & 660.34 \\
\texttt{B8\_64\_10} & 16 & 30 & 14 &  & 5/5 & 109.69 & 660.28 \\
\midrule
\texttt{B9\_81\_1} & 17 & 33 & 16 &  & 5/5 & 492.89 & 660.20 \\
\texttt{B9\_81\_2} & 16 & 34 & 18 &  & 5/5 & 303.16 & 660.20 \\
\texttt{B9\_81\_3} & 16 & 35 & 19 &  & 5/5 & 0.52 & 660.22 \\
\texttt{B9\_81\_4} & 16 & 34 & 18 &  & 3/5 & 306.83 & 660.20 \\
\texttt{B9\_81\_5} & 17 & 34 & 17 &  & 1/5 & 114.96 & 660.21 \\
\texttt{B9\_81\_6} & 17 & 34 & 17 &  & 2/5 & 223.07 & 660.22 \\
\texttt{B9\_81\_7} & 16 & 34 & 18 &  & 5/5 & 463.06 & 660.20 \\
\texttt{B9\_81\_8} & 16 & 34 & 18 &  & 5/5 & 17.41 & 660.17 \\
\texttt{B9\_81\_9} & 16 & 34 & 18 &  & 2/5 & 223.93 & 660.16 \\
\texttt{B9\_81\_10} & 16 & 34 & 18 &  & 5/5 & 298.55 & 660.17 \\
\midrule
\texttt{B10\_100\_1} & 18 & 47 & 29 &  & 5/5 & 260.02 & 660.13 \\
\texttt{B10\_100\_2} & 18 & 46 & 28 &  & 5/5 & 286.91 & 660.21 \\
\texttt{B10\_100\_3} & 18 & 46 & 28 &  & 5/5 & 394.61 & 660.17 \\
\texttt{B10\_100\_4} & 18 & 47 & 29 &  & 5/5 & 89.10 & 660.15 \\
\texttt{B10\_100\_5} & 18 & 47 & 29 &  & 5/5 & 258.41 & 660.13 \\
\texttt{B10\_100\_6} & 18 & 47 & 29 &  & 5/5 & 33.92 & 660.11 \\
\texttt{B10\_100\_7} & 18 & 47 & 29 &  & 5/5 & 258.58 & 660.18 \\
\texttt{B10\_100\_8} & 18 & 46 & 28 &  & 1/5 & 222.96 & 660.14 \\
\texttt{B10\_100\_9} & 18 & 46 & 28 &  & 1/5 & 115.04 & 660.15 \\
\texttt{B10\_100\_10} & 18 & 46 & 28 &  & 2/5 & 220.95 & 660.19 \\
\midrule
\texttt{B11\_121\_1} & 20 & 52 & 32 &  & 3/5 & 184.92 & 660.05 \\
\texttt{B11\_121\_2} & 20 & 51 & 31 &  & 5/5 & 362.53 & 660.06 \\
\texttt{B11\_121\_3} & 20 & 52 & 32 &  & 5/5 & 284.82 & 660.07 \\
\texttt{B11\_121\_4} & 20 & 52 & 32 &  & 5/5 & 276.81 & 660.06 \\
\texttt{B11\_121\_5} & 20 & 52 & 32 &  & 5/5 & 271.28 & 660.07 \\
\texttt{B11\_121\_6} & 20 & 52 & 32 &  & 5/5 & 210.28 & 660.06 \\
\texttt{B11\_121\_7} & 20 & 51 & 31 &  & 1/5 & 147.58 & 660.06 \\
\texttt{B11\_121\_8} & 20 & 51 & 31 &  & 5/5 & 356.65 & 660.06 \\
\texttt{B11\_121\_9} & 20 & 52 & 32 &  & 5/5 & 110.80 & 660.06 \\
\texttt{B11\_121\_10} & 20 & 52 & 32 &  & 5/5 & 29.26 & 660.05 \\
\midrule
\texttt{B12\_144\_1} & 21 & 56 & 35 &  & 5/5 & 149.11 & 660.08 \\
\texttt{B12\_144\_2} & 21 & 55 & 34 &  & 4/5 & 446.21 & 660.08 \\
\texttt{B12\_144\_3} & 21 & 56 & 35 &  & 1/5 & 256.07 & 660.08 \\
\texttt{B12\_144\_4} & 21 & 56 & 35 &  & 5/5 & 489.77 & 660.06 \\
\texttt{B12\_144\_5} & 22 & 56 & 34 &  & 2/5 & 292.10 & 660.09 \\
\texttt{B12\_144\_6} & 21 & 55 & 34 &  & 3/5 & 445.57 & 660.08 \\
\texttt{B12\_144\_7} & 22 & 56 & 34 &  & 2/5 & 171.15 & 660.08 \\
\texttt{B12\_144\_8} & 22 & 56 & 34 &  & 2/5 & 234.70 & 660.06 \\
\texttt{B12\_144\_9} & 21 & 56 & 35 &  & 5/5 & 398.83 & 660.07 \\
\texttt{B12\_144\_10} & 21 & 56 & 35 &  & 2/5 & 229.11 & 660.08 \\
\midrule
\texttt{B13\_169\_1} & 23 & 60 & 37 &  & 2/5 & 328.84 & 660.03 \\
\texttt{B13\_169\_2} & 23 & 61 & 38 &  & 4/5 & 191.33 & 660.03 \\
\texttt{B13\_169\_3} & 23 & 61 & 38 &  & 5/5 & 108.83 & 660.03 \\
\texttt{B13\_169\_4} & 23 & 61 & 38 &  & 5/5 & 330.51 & 660.04 \\
\texttt{B13\_169\_5} & 23 & 61 & 38 &  & 4/5 & 166.78 & 660.03 \\
\texttt{B13\_169\_6} & 23 & 60 & 37 &  & 2/5 & 288.65 & 660.04 \\
\texttt{B13\_169\_7} & 23 & 61 & 38 &  & 5/5 & 97.62 & 660.04 \\
\texttt{B13\_169\_8} & 23 & 61 & 38 &  & 5/5 & 139.86 & 660.04 \\
\texttt{B13\_169\_9} & 23 & 61 & 38 &  & 5/5 & 248.55 & 660.03 \\
\texttt{B13\_169\_10} & 23 & 61 & 38 &  & 4/5 & 208.10 & 660.03 \\
\midrule
\texttt{B14\_196\_1} & 24 & 65 & 41 &  & 4/5 & 146.87 & 660.04 \\
\texttt{B14\_196\_2} & 25 & 64 & 39 &  & 1/5 & 103.00 & 660.05 \\
\texttt{B14\_196\_3} & 25 & 65 & 40 &  & 2/5 & 188.98 & 660.04 \\
\texttt{B14\_196\_4} & 25 & 66 & 41 &  & 5/5 & 216.12 & 660.03 \\
\texttt{B14\_196\_5} & 25 & 65 & 40 &  & 5/5 & 200.86 & 660.04 \\
\texttt{B14\_196\_6} & 24 & 65 & 41 &  & 5/5 & 407.17 & 660.04 \\
\texttt{B14\_196\_7} & 25 & 64 & 39 &  & 1/5 & 304.17 & 660.03 \\
\texttt{B14\_196\_8} & 24 & 65 & 41 &  & 3/5 & 185.06 & 660.04 \\
\texttt{B14\_196\_9} & 24 & 65 & 41 &  & 5/5 & 256.88 & 660.04 \\
\texttt{B14\_196\_10} & 25 & 65 & 40 &  & 5/5 & 469.19 & 660.04 \\
\midrule
\texttt{B15\_225\_1} & 26 & 68 & 42 &  & 3/5 & 428.43 & 660.03 \\
\texttt{B15\_225\_2} & 26 & 68 & 42 &  & 1/5 & 121.00 & 660.02 \\
\texttt{B15\_225\_3} & 26 & 68 & 42 &  & 1/5 & 324.21 & 660.03 \\
\texttt{B15\_225\_4} & 26 & 68 & 42 &  & 1/5 & 393.84 & 660.05 \\
\texttt{B15\_225\_5} & 26 & 69 & 43 &  & 1/5 & 52.81 & 660.03 \\
\texttt{B15\_225\_6} & 26 & 68 & 42 &  & 2/5 & 420.71 & 660.03 \\
\texttt{B15\_225\_7} & 26 & 68 & 42 &  & 1/5 & 210.44 & 660.04 \\
\texttt{B15\_225\_8} & 26 & 68 & 42 &  & 3/5 & 332.92 & 660.04 \\
\texttt{B15\_225\_9} & 26 & 68 & 42 &  & 1/5 & 189.16 & 660.03 \\
\texttt{B15\_225\_10} & 26 & 68 & 42 &  & 1/5 & 240.63 & 660.05 \\
\midrule
\texttt{B16\_256\_1} & 28 & 73 & 45 &  & 5/5 & 30.30 & 660.02 \\
\texttt{B16\_256\_2} & 28 & 73 & 45 &  & 5/5 & 39.34 & 660.02 \\
\texttt{B16\_256\_3} & 28 & 73 & 45 &  & 1/5 & 72.98 & 660.03 \\
\texttt{B16\_256\_4} & 27 & 73 & 46 &  & 3/5 & 94.18 & 660.03 \\
\texttt{B16\_256\_5} & 28 & 73 & 45 &  & 2/5 & 150.77 & 660.03 \\
\texttt{B16\_256\_6} & 28 & 73 & 45 &  & 4/5 & 146.74 & 660.02 \\
\texttt{B16\_256\_7} & 27 & 72 & 45 &  & 1/5 & 77.44 & 660.03 \\
\texttt{B16\_256\_8} & 28 & 72 & 44 &  & 2/5 & 135.27 & 660.02 \\
\texttt{B16\_256\_9} & 27 & 73 & 46 &  & 5/5 & 22.90 & 660.02 \\
\texttt{B16\_256\_10} & 28 & 72 & 44 &  & 1/5 & 143.66 & 660.02 \\
\midrule
\texttt{B24\_576\_1} & 39 & 98 & 59 &  & 1/5 & 204.10 & 660.02 \\
\texttt{B24\_576\_2} & 39 & 98 & 59 &  & 1/5 & 302.48 & 660.02 \\
\texttt{B24\_576\_3} & 39 & 98 & 59 &  & 1/5 & 25.66 & 660.03 \\
\texttt{B24\_576\_4} & 39 & 98 & 59 &  & 1/5 & 123.99 & 660.03 \\
\texttt{B24\_576\_5} & 39 & 98 & 59 &  & 1/5 & 35.49 & 660.02 \\
\texttt{B24\_576\_6} & 39 & 98 & 59 &  & 1/5 & 110.62 & 660.02 \\
\texttt{B24\_576\_7} & 39 & 98 & 59 &  & 1/5 & 82.90 & 660.02 \\
\texttt{B24\_576\_8} & 39 & 98 & 59 &  & 2/5 & 62.60 & 660.04 \\
\texttt{B24\_576\_9} & 39 & 98 & 59 &  & 1/5 & 21.14 & 660.03 \\
\texttt{B24\_576\_10} & 39 & 98 & 59 &  & 1/5 & 43.98 & 660.03 \\
\midrule
\texttt{B32\_1024\_1} & 53 & 150 & 97 &  & 2/5 & 268.95 & 660.03 \\
\texttt{B32\_1024\_2} & 53 & 151 & 98 &  & 3/5 & 98.52 & 660.03 \\
\texttt{B32\_1024\_3} & 53 & 149 & 96 &  & 2/5 & 222.22 & 660.03 \\
\texttt{B32\_1024\_4} & 53 & 150 & 97 &  & 2/5 & 59.44 & 660.02 \\
\texttt{B32\_1024\_5} & 53 & 150 & 97 &  & 2/5 & 67.92 & 660.03 \\
\texttt{B32\_1024\_6} & 53 & 149 & 96 &  & 1/5 & 148.52 & 660.02 \\
\texttt{B32\_1024\_7} & 53 & 150 & 97 &  & 1/5 & 211.46 & 660.02 \\
\texttt{B32\_1024\_8} & 53 & 151 & 98 &  & 2/5 & 76.30 & 660.03 \\
\texttt{B32\_1024\_9} & 53 & 150 & 97 &  & 3/5 & 70.19 & 660.02 \\
\texttt{B32\_1024\_10} & 53 & 150 & 97 &  & 3/5 & 182.31 & 660.04 \\
\midrule
\texttt{B48\_2304\_1} & 75 & 200 & 125 &  & 2/5 & 50.12 & 660.03 \\
\texttt{B48\_2304\_2} & 75 & 199 & 124 &  & 2/5 & 307.22 & 660.03 \\
\texttt{B48\_2304\_3} & 75 & 199 & 124 &  & 1/5 & 239.91 & 660.03 \\
\texttt{B48\_2304\_4} & 75 & 200 & 125 &  & 4/5 & 224.27 & 660.02 \\
\texttt{B48\_2304\_5} & 75 & 200 & 125 &  & 3/5 & 158.18 & 660.02 \\
\texttt{B48\_2304\_6} & 75 & 199 & 124 &  & 2/5 & 93.39 & 660.05 \\
\texttt{B48\_2304\_7} & 75 & 199 & 124 &  & 4/5 & 179.57 & 660.03 \\
\texttt{B48\_2304\_8} & 75 & 200 & 125 &  & 4/5 & 73.63 & 660.03 \\
\texttt{B48\_2304\_9} & 75 & 199 & 124 &  & 1/5 & 80.26 & 660.03 \\
\texttt{B48\_2304\_10} & 75 & 200 & 125 &  & 2/5 & 84.84 & 660.03 \\
\midrule
\texttt{B64\_4096\_1} & 97 & 244 & 147 &  & 4/5 & 125.86 & 660.04 \\
\texttt{B64\_4096\_2} & 97 & 242 & 145 &  & 1/5 & 74.30 & 660.04 \\
\texttt{B64\_4096\_3} & 97 & 243 & 146 &  & 4/5 & 380.74 & 660.03 \\
\texttt{B64\_4096\_4} & 97 & 243 & 146 &  & 2/5 & 222.72 & 660.04 \\
\texttt{B64\_4096\_5} & 96 & 243 & 147 &  & 2/5 & 79.33 & 660.04 \\
\texttt{B64\_4096\_6} & 97 & 243 & 146 &  & 1/5 & 240.90 & 660.04 \\
\texttt{B64\_4096\_7} & 97 & 244 & 147 &  & 2/5 & 62.80 & 660.04 \\
\texttt{B64\_4096\_8} & 97 & 243 & 146 &  & 1/5 & 18.83 & 660.04 \\
\texttt{B64\_4096\_9} & 97 & 242 & 145 &  & 1/5 & 143.71 & 660.04 \\
\texttt{B64\_4096\_10} & 97 & 243 & 146 &  & 1/5 & 218.36 & 660.04 \\
\midrule
\texttt{B100\_10000\_1} & 151 & 404 & 253 &  & 4/5 & 105.39 & 660.21 \\
\texttt{B100\_10000\_2} & 151 & 406 & 255 &  & 2/5 & 92.97 & 660.21 \\
\texttt{B100\_10000\_3} & 151 & 404 & 253 &  & 2/5 & 156.23 & 660.19 \\
\texttt{B100\_10000\_4} & 151 & 404 & 253 &  & 1/5 & 141.99 & 660.21 \\
\texttt{B100\_10000\_5} & 151 & 403 & 252 &  & 3/5 & 263.30 & 660.26 \\
\bottomrule
\end{longtable}}

\endgroup

\end{document}